\documentclass[a4paper]{article}
\pdfoutput=1

\RequirePackage[OT1]{fontenc}
\RequirePackage{amsthm,amsmath,natbib}
\RequirePackage[colorlinks,citecolor=blue,urlcolor=blue]{hyperref}

\usepackage{amsthm}
\usepackage{amsmath}
\usepackage{amssymb}
\usepackage{geometry}
\usepackage{titlesec}
\usepackage[titles]{tocloft}
\usepackage{epsfig}
\usepackage{epstopdf}
\usepackage{multirow}
\usepackage[font=small,labelfont=bf]{caption}
\usepackage{enumerate}
\usepackage{dsfont}
\usepackage{tikz}
\usepackage{pgfplots}
\usepackage{amssymb}
\usepackage{bbm}
\usepackage{subfig}
\usepackage{graphicx}
\usepackage{float}
\usepackage{caption}
\usepackage{indentfirst}
\usepackage{xspace}

\newtheorem{assumption}{Assumption}[section]
\newtheorem{theorem}{Theorem}[section]
\newtheorem{lemma}{Lemma}[section]

\newtheorem{corollary}{Corollary}[section]
\theoremstyle{plain}

\newcommand{\bdelta}{{\boldsymbol{\delta}}}

\newcommand{\bV}{{\boldsymbol{V}}}
\newcommand{\bv}{{\boldsymbol{v}}}

\newcommand{\bY}{{\boldsymbol{Y}}}
\newcommand{\by}{{\boldsymbol{y}}}

\newcommand{\bu}{{\boldsymbol{u}}}
\newcommand{\bU}{{\boldsymbol{U}}}

\newcommand{\bG}{{\boldsymbol{G}}}

\newcommand{\btheta}{{\boldsymbol{\theta}}}

\numberwithin{equation}{section}
\theoremstyle{plain}
\newtheorem{remark}{Remark}[section]
\usepackage{authblk}

\title{Nonparametric generalized fiducial inference for survival functions under censoring}
\author[1]{Yifan Cui}
\author[1]{Jan Hannig}
\affil[1]{Department of Statistics and Operations Research, University of North Carolina, Chapel Hill, NC 27599}

\usepackage{pdfpages}

\begin{document}
\maketitle








\begin{abstract}
Fiducial Inference, introduced by Fisher in the 1930s, has a long history, which at times aroused passionate disagreements. However, its application has been largely confined to relatively simple parametric problems. In this paper, we present what might be the first time fiducial inference, as generalized by \citet{hannig2016generalized},  is systematically applied to estimation of a nonparametric survival function under right censoring.  We find that the resulting fiducial distribution gives rise to surprisingly good statistical procedures applicable to both one sample and two sample problems. In particular, we use the fiducial distribution of a survival function to construct pointwise and curvewise confidence intervals for the survival function, and propose tests based on the curvewise confidence interval. We establish a functional Bernstein-von Mises theorem, and perform thorough simulation studies in scenarios with different levels of censoring. The proposed fiducial based confidence intervals maintain coverage in situations where asymptotic methods often have substantial coverage problems. Furthermore, the average length of the proposed confidence intervals is often shorter than the length of competing methods that maintain coverage. Finally, the proposed fiducial test is more powerful than various types of log-rank tests and sup log-rank tests in some scenarios. We illustrate the proposed fiducial test comparing chemotherapy against chemotherapy combined with radiotherapy using data from the treatment of locally unresectable gastric cancer.
\end{abstract}

\smallskip
\noindent \textbf{Keywords:} Generalized fiducial inference, Right censored data, Nonparametric model, Coverage, Testing.



\section{Introduction}

Fiducial inference can be traced back to a series of articles by the father of modern statistics R.~A.~\cite{Fisher1925,  Fisher1930, Fisher1933, Fisher1935a} who introduced the concept 
 as a potential replacement of the Bayesian posterior distribution. A systematic development of the idea has been hampered by ambiguity, as \cite{Brillinger1962examples} describes:
``The reason for this lack of agreement and the resulting controversy is possibly
due to the fact that the fiducial method has been put forward as a general logical principle, but yet has been illustrated mainly by means of particular examples rather than broad requirements.''
Indeed, we contend that until recently fiducial inference was applied to relatively a small class of parametric problems only.

Since the mid 2000s, there has been a renewed interest in modifications of fiducial inference.
\cite{Hannig2009, Hannig2013} bring forward a mathematical definition of what they call the Generalized Fiducial Distribution (GFD).  Having a formal definition allowed fiducial inference to be applied to a wide variety of statistical settings \citep{HannigIyerWang2007, WangIyer2005, WangIyer2006a, WangIyer2006b,WangHannigIyer2012b, HannigIyerQingWangLiu2018, CisewskiHannig2012, WandlerHannig2012b,HannigLee2009,LaiHannigLee2014,liu2017generalized}.

Other related approaches include Dempster-Shafer theory \citep{Dempster2008, EdlefsenLiuDempster2009}, inferential models \citep{martin2015inferential}, and
 confidence distributions \citep{XieSingh2013, schweder2016confidence, HjortSchweder2018}.  Objective Bayesian inference, which aims at finding non-subjective model based priors can also be seen as addressing the same basic question. Examples of recent breakthroughs related to reference prior and model selection are  \citet{BayarriEtAl2012, BergerBernardoSun2009, BergerBernardoSun2012}. There are many more references that interested readers can find in the review article \cite{hannig2016generalized}.


In this paper, we apply the fiducial approach in the context of survival analysis. To our knowledge, this is the first time fiducial inference has been systematically applied to an infinite-dimensional statistical problem. However, for use of confidence distributions to address some basic non-parametric problems see Chapter 11 of \cite{schweder2016confidence}.
In this manuscript, we propose a computationally efficient algorithm to sample from the GFD, and use the samples from the GFD to construct statistical procedures. The median of the GFD could be considered as a substitution for the Kaplan-Meier estimator \citep{kaplan1958nonparametric}, which is a classical estimator in survival analysis. Appropriate quantiles of the GFD evaluated at a given time provide pointwise confidence intervals for survival function. Similarly, the confidence intervals for quantiles of survival functions can be obtained by inverting the GFD.

The proposed pointwise confidence intervals maintain coverage in situations where classical confidence intervals often have coverage problems \citep{fay2013pointwise}. \cite{fay2013pointwise,fay2016finite} construct solutions to avoid these coverage problems.
 It is interesting to note that the conservative version of the proposed pointwise fiducial confidence interval is equivalent to beta product confidence procedure confidence interval of \cite{fay2013pointwise}. The other fiducial confidence interval proposed in this paper is based on log-linear interpolation and has the shortest length among all existing methods which maintain coverage.

We also construct curvewise confidence intervals for survival functions. Based on the curvewise confidence intervals, we propose a two sample test for testing whether two survival functions are equal. The proposed test does not need the proportional hazard assumption \citep{bouliotis2011crossing}, and appears to be a good replacement for the log-rank test and sup log-rank test. 

We establish an asymptotic theory which verifies the frequentist validity of the proposed fiducial approach. In particular, we prove a functional Bernstein--von Mises theorem for the GFD in Skorokhod's $D[0,t]$ space. Because randomness in GFD comes from two distinct sources the proof of this results is different from the usual proof of asymptotic normality for the Kaplan-Meyer estimator. As a consequence of the functional Bernstein--von Mises theorem, the proposed  pointwise and curvewise confidence intervals provide asymptotically correct coverage, and the proposed survival function estimator is asymptotically equivalent to the Kaplan-Meier estimator. 

We report results of a simulation study showing the proposed fiducial methods provide competitive, and in some cases superior performance to the methods in the literature. In particular, we compare the performance of the GFD intervals with classical confidence intervals like Greenwood \citep{survival-package}, Borkowf \citep{borkowf2005simple}, Strawderman-Wells \citep{strawderman1997accurate,strawderman1997accurateb}, nonparametric bootstrap \citep{efron1981censored,akritas1986bootstrapping}, constrained bootstrap \citep{barber1999symmetric}, Thomas-Grunkemeier method \citep{thomas1975confidence}, constrained beta \citep{barber1999symmetric}, and beta product confidence procedure \citep{fay2013pointwise,fay2016finite} in various settings with small samples and/or heavy censoring. Additionally we also consider the setting of \cite{barber1999symmetric} in which the data contains fewer censored observations. Next, we report several scenarios showing the desirable power of the GFD test in comparison to 12 different types of log-rank tests implemented in the R package survMisc \citep{dardis2016survmisc}: original log-rank test \citep{mantel1966evaluation}; Gehan-Breslow generalized Wilcoxon log-rank test \citep{gehan1965generalized}; Tarone-Ware log-rank test \citep{tarone1977distribution}; Peto-Peto log-rank test \citep{peto1972asymptotically}; Modified Peto-Peto log-rank test \citep{andersen1982cox}; Fleming-Harrington log-rank test \citep{harrington1982class} and corresponding supremum versions \citep{fleming1987supremum,eng2005sample}.

We apply the proposed fiducial method to test the difference between chemotherapy and chemotherapy combined with radiotherapy in the treatment of locally unresectable gastric cancer \citep{klein2005survival}. The proposed fiducial test has the smallest p-value compared to existing methods. We also report a small simulation study based on 500 synthetic datasets mimicking the cancer data. The proposed fiducial test is more powerful than the 12 different tests described above.


\section{Methodology}\label{s:Methodology}

\subsection{Fiducial approach explained}\label{nonsur}
In this section, we explain the definition of a generalized fiducial distribution. We demonstrate the definition on the problem of estimating survival functions when no censoring is present.
We start by expressing the relationship between the data $\bY$ and the parameter $\btheta$ using
\begin{equation}\label{eq:StructuralEq}
 \bY = \bG(\bU,\btheta),
\end{equation}
where $\bG(\cdot,\cdot)$ is a deterministic function termed the data generating equation, and $\bU$ is a random vector whose distribution is independent of $\btheta$ and completely known.
Data $\bY$ could be simulated by generating a random variable $\bU$ and plugging it into the data generating equation \eqref{eq:StructuralEq}.
For example, a data generating equation for the $N(\mu,\sigma^2)$ model is  $Y_i=G(U_i,\mu,\sigma)=\mu + \sigma \Phi^{-1}(U_i)$, where $\bU=(U_1,\cdots,U_n)$ are independent and identically distributed $U(0,1)$ and $\Phi(y)$ is the distribution function of the standard normal distribution.

The inverse cumulative distribution function method for generating random variables provides a common data generating equation for a nonparametric independent and identically distributed model:
\begin{equation}\label{eq:StructuralEq1}
  Y_i = G(U_i,F) = F^{-1} (U_i),\quad i=1, \ldots, n,
\end{equation}
where  $F^{-1}(u)=\inf\{y\in \mathbbm{R}: F(y)\geq u\}$ is the usual ``inverse'' of the distribution function $F(y)$\citep{CasellaBerger2002}. Notice that the distribution function $F$ itself is the parameter $\btheta$ in this infinite dimensional model. The actual observed data is generated using the true distribution function $F_0$.

Roughly speaking a GFD is obtained by inverting the data generating equation, and \cite{hannig2016generalized} proposes a very general definition of GFD. However, in order to simplify the presentation, we will use an earlier, less general version found in \cite{Hannig2009}. The two definitions are equivalent for the models considered here.

We start by denoting the inverse image of the data generating equation \eqref{eq:StructuralEq} by
\[
 Q(\by,\bu)=\{\btheta\,:\,  \by = \bG(\bu,\btheta)\}.
\]
For the special case \eqref{eq:StructuralEq1} the inverse image  is
\begin{equation}\label{eq:QP1}
Q(\by,\bu)=\bigcap_{i=1}^n\{F: F(y_i)\geq u_i, F(y_i-\epsilon)< u_i ~\text{for any}~ \epsilon>0\}.
\end{equation}

If $\by$ is the observed data and $\bu_0$ the value of the random vector $\bU$ that was used to generate it, then we are guaranteed that the true parameter value $\btheta_0\in Q(\by,\bu_0)$. However, we only know a distribution of $\bU$ and not the actual value $\bu_0$. Notice that $\by=\bG(\bu_0,\btheta_0)$ and therefore only values of $\bu$ for which $Q(\by,\bu)\neq\emptyset$ should be considered. Let $\bU^*$ be another random variable independent of and having the same distribution as $\bU$.  Since the conditional distribution of
$\bU^* \mid  \{ Q({\by},\bU^*)\neq\emptyset\}$ can be viewed as summarizing our knowledge about $\bu_0$, the conditional distribution of
\begin{equation}\label{eq:FiducialDef}
  Q(\by,\bU^*) \mid \{ Q({\by},\bU^*)\neq\emptyset\}
\end{equation}
can be viewed as summarizing our knowledge about $\btheta_0$.

Notice that $Q(\by,\bu)$ is a set that can contain more than one element. We deal with this by selecting a representative from the closure of  $Q(\by,\bu)$. The distribution of a representative selected from \eqref{eq:FiducialDef} is a GFD.
Based on the theoretical results presented, the non-uniqueness caused by this somewhat arbitrary choice disappears asymptotically. A possible conservative alternative to selecting a single representative from $Q(\by,\bu)$ could use the theory of belief functions \citep{Dempster2008, shafer1976mathematical}.

To describe the GFD in the particular case of \eqref{eq:StructuralEq1} we define for all $s\geq 0$, $F^L_{(\by,\bu)}(s)=\inf\{F(s):F\in Q(\by,\bu)\}$ and $F^U_{(\by,\bu)}(s)=\sup\{F(s):F\in Q(\by,\bu)\}$. The closure of the inverse image \eqref{eq:QP1} is a set of all distribution functions $F$ that stay between $F^L_{(\by,\bu)}$ and $F^U_{(\by,\bu)}$.  Also notice that $Q(\by,\bu)$ is not empty if and only if the order of $\bu$ matches the order of $\by$, with the understanding that in the case of ties in $\by$, the $u_i$'s corresponding to the ties could be any order.

By exchangeability, the conditional distribution $\bU^* \mid \{Q(\by,\bU^*)\neq\emptyset\}$ is the same as  the distribution of $\bU^*_{[\by]}$, where $\bU^*_{[\by]}$ is the independent and identically distributed U(0,1) reordered to match the order of $\by$. Thus, any distribution stochastically larger than $F^L_{(\by,\bU^*_{[\by]})}$ and stochastically smaller than $F^U_{(\by,\bU^*_{[\by]})}$ is a GFD. 
Sampling from this fiducial distribution is easy to implement. 

We consider the following 2 main options in using the GFD for inference. The first option is to construct conservative confidence sets. For example, when designing pointwise confidence intervals for the survival function at time $s$,  
we use quantiles of $1-F^U_{(\by,\bU^*_{[\by]})}(s)$  for lower bounds and  quantiles of $1-F^L_{(\by,\bU^*_{[\by]})}(s)$  for upper bounds.

The second option is to select a suitable representative of  $Q{(\by,\bU^*_{[\by]})}$. When there are no ties present in the data we propose to fit a continuous distribution function by using linear interpolation for the survival function on the log scale, i.e., the distribution function $F^I_{(\by,\bu)}(s)=1-e^{L(s)}$, where $L(s)$ is the linear interpolation between $(0,0),(y_{(1)},\log u_{(1)}),\ldots,(y_{(n)},\log u_{(n)})$, and on the interval $(y_{(n)},\infty)$ we extrapolate by extending the line between $(y_{(n-1)},\log u_{(n-1)})$ and $(y_{(n)},\log u_{(n)})$. We will call this the log-linear interpolation.

As usually, we denote the GFD for survival functions $S^L_{(\by,\bu)}=1-F^U_{(\by,\bu)}$, $S^U_{(\by,\bu)}=1-F^L_{(\by,\bu)}$, and $S^I_{(\by,\bu)}=1-F^I_{(\by,\bu)}$. For simplicity, hereinafter we omit the subindex $(\by,\bu)$.
In the rest of this paper we will also denote Monte Carlo samples of the lower bound, the upper bound, and the log-linear interpolation of the GFD for the survival function  by $S_i^L,S_i^U,$ and $S_i^I~(i=1,\ldots, m)$, respectively.

To demonstrate the fiducial distribution of this section, we draw 300 observations from $Weibull(20,10)$. 
Based on this data, we plot a fiducial sample of survival functions $S^I_i (i=1,\ldots,1000$) and the empirical survival function in the left panel of Figure \ref{both}.

\subsection{Fiducial approach in survival setting}\label{sur}
In this section, we derive the GFD for the failure distribution based on right censored data. Here we treat the situation when the failure and censoring times are independent. The same GFD is derived under a more general model that includes dependence between failure and censoring times in the Appendix.

Let failure times $X_i~(i=1,\ldots,n)$ follow the true  distribution function $F_0$ and censoring times $Z_i~(i=1,\ldots,n)$ have the distribution function $R_0$. We observe partially censored data $\{y_i,\delta_i\}$ $(i=1,\ldots n)$, where $y_i=x_i\wedge z_i$ is the minimum of $x_i$ and $z_i$, $\delta_i=I\{x_i\leq z_i\}$ denotes censoring indicator.

We consider the following data generating equation,
\begin{align}\label{eq:newStructuralEq}
Y_i=F^{-1}(U_i)\wedge R^{-1}(V_i),\quad \delta_i=I\{F^{-1}(U_i)\leq R^{-1}(V_i)\},
\end{align}
where $U_i, V_i$ are independent and identically distributed $U(0,1)$ and the actual observed data were generated using $F=F_0$ and $R=R_0$. We are committing a slight abuse of notation as $\bY$ in Equation~\eqref{eq:StructuralEq} is $(\bY,\bdelta)$ in Equation~\eqref{eq:newStructuralEq} and $\bU$ in Equation~\eqref{eq:StructuralEq} is $(\bU,\bV)$ in Equation~\eqref{eq:newStructuralEq}.

For a failure event $\delta_i=1$, we have full information about failure time $x_i$, i.e., $x_i=y_i$, and partial information about censoring time $z_i$, i.e., $z_i\geq y_i$. In this case, just as in the previous section,
\begin{align*}
F^{-1}(u_i)=y_i \quad & \text{if and only if} \quad  F(y_i)\geq u_i, F(y_i-\epsilon)< u_i ~\text{for any}~ \epsilon>0.
\end{align*}

For a censored event $\delta_i=0$, we know only partial information about $x_i$, i.e., $x_i > y_i$, and full information on $z_i$, i.e., $z_i= y_i$. Similarly,
\begin{align*}
F^{-1}(u_i)> y_i \quad & \text{if and only if} \quad F(y_i)< u_i,\\
R^{-1}(v_i)=y_i \quad & \text{if and only if} \quad  R(y_i)\geq v_i, R(y_i-\epsilon)< v_i ~\text{for any}~ \epsilon>0.
\end{align*}

To obtain the inverse map, we start by inverting a single observation. If $\delta_i=1$, the inverse map for this datum is
\[
Q^{F,R}_1(y_i,u_i,v_i)
=\{F: F(y_i)\geq u_i, F(y_i-\epsilon)< u_i ~\text{for any}~ \epsilon>0\}\times\{R:R^{-1}(v_i)\geq y_i\}.
\]
If $\delta_i=0$, the inverse map is 
\[
Q^{F,R}_0(y_i,u_i,v_i)
=\{F:F(y_i)< u_i\}\times \{R: R(y_i)\geq v_i, R(y_i-\epsilon)< v_i ~\text{for any}~ \epsilon>0\}.
\]
Combining these  we obtain the complete inverse map
\begin{equation}\label{eq:QP2}
Q^{F,R}(\by,\bdelta,\bu,\bv)=\bigcap_i Q^{F,R}_{\delta_i}(y_i,u_i,v_i)=
Q^{F}(\by,\bdelta,\bu)\times Q^{R}(\by,\bdelta,\bv),
\end{equation}
where
\begin{equation}\label{eq:QP2F}
Q^F(\by,\bdelta,\bu)=\left\{F:
\begin{cases}
  F(y_i)\geq u_i,  F(y_i-\epsilon)< u_i ~\text{for any}~ \epsilon>0 & \mbox{for all  $i$ such that $\delta_i=1$}\\
  F(y_j)< u_j                         &  \mbox{for all  $j$ such that $\delta_j=0$}
 \end{cases}
  \right\} ,
\end{equation}
and $Q^R(\by,\bdelta,\bv)$ is analogous. Notice that the inverse of $Q^{F,R}$ in (6) is in the form of a Cartesian product. This is a direct consequence of our choice of data generating equation, and it greatly simplifies the calculation of marginal fiducial distribution for failure times.

To demonstrate the inverse \eqref{eq:QP2F}, Figure~\ref{explain} presents the survival function representation of $Q^F(\by,\bdelta,\bu)$ for one small data set  $(n=8)$ of  $X\sim Weibull(20,10)$ censored by $Z\sim Exp(20)$, and two different values of $\bu$.  The circle points denote failure observations and the triangle points denote censored observations. Any survival function lying between the upper and the lower bounds is an element of the closure of $Q^F(\by,\bdelta,\bu)$. In particular, we plot the log-linear interpolation going through the failure observations as described in Section~\ref{nonsur} with a modification to ensure it satisfies the lower fiducial bound.  Notice that the upper fiducial bound changes at the failure times only, while the lower fiducial bound changes at all failure times and at some censoring times depending on the value of $\bu$. 
\begin{figure}[t]\centering
\includegraphics[width=60mm]{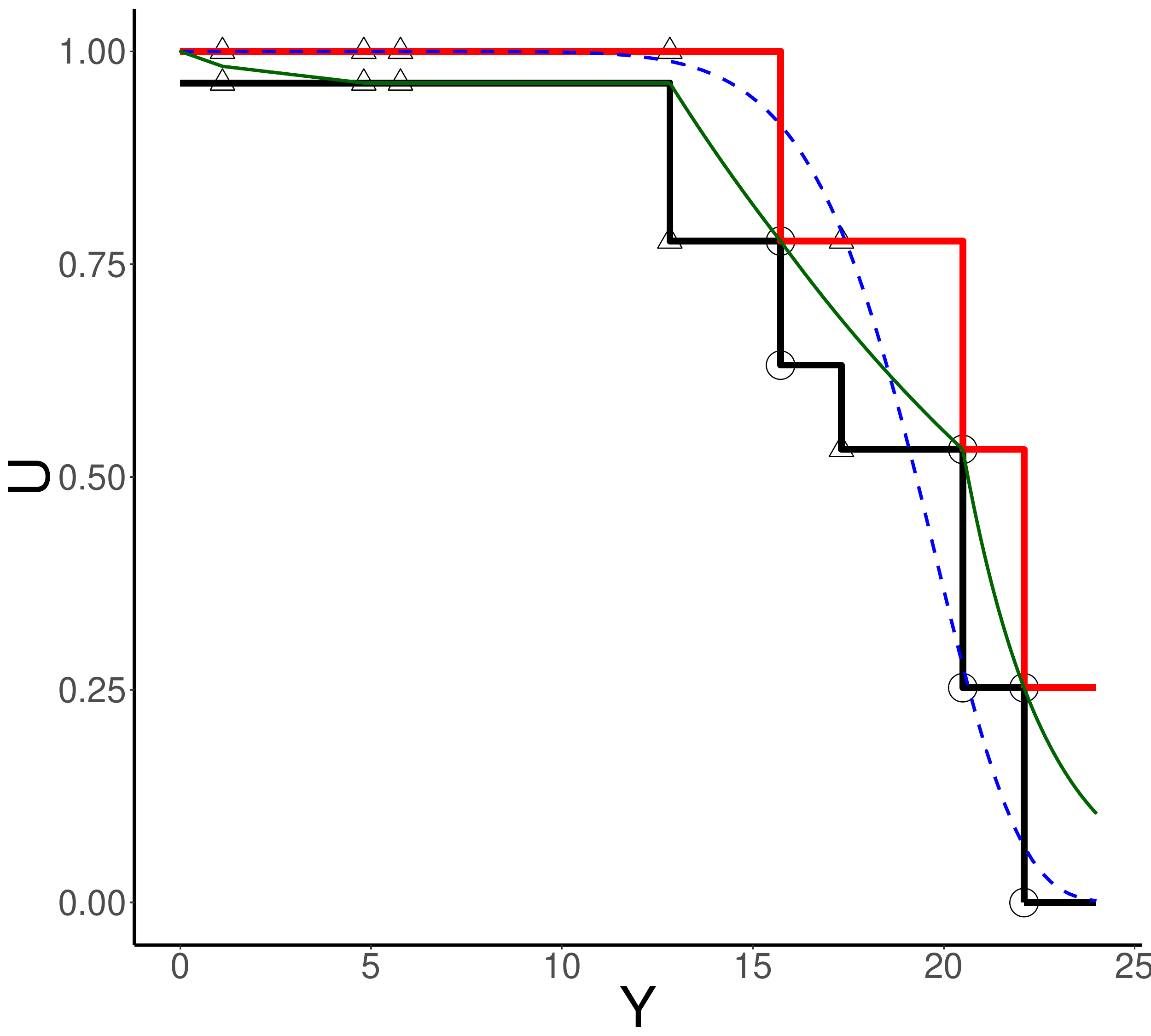}
\includegraphics[width=60mm]{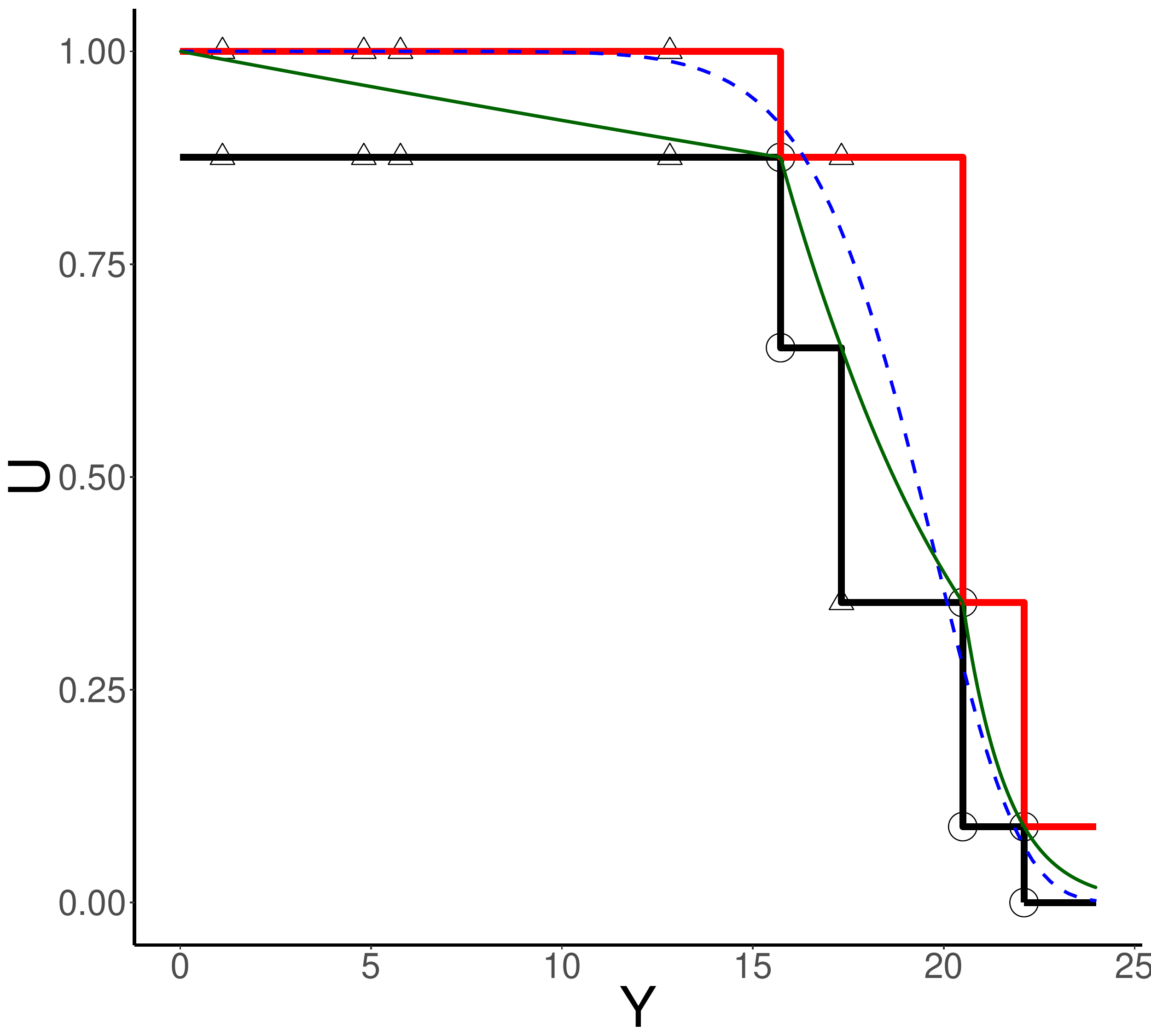}
\caption{Two realizations of fiducial curves for a sample of size $8$ from $Weibull(20,10)$ censored by $Exp(20)$. Here fiducial curves refer to Monte Carlo samples $S^L_i$, $S^U_i$, and $S^I_i$ ($i=1,2$) from the GFD.} The red curve is an upper bound and the black curve is a lower bound. The green curve is the log-linear interpolation. The circle points denote failure observations. The triangle points denote censored observations. The dashed blue curve is the true survival function of $Weibull(20,10)$. Since the fiducial distribution reflects uncertainty we do not expect every fiducial curve to be close to the true survival function.
\label{explain}
\end{figure}

When defining the GFD, let $(\bU^*,\bV^*)$ be independent of and having the same distribution as $(\bU,\bV)$.
Because of  the way the inverse \eqref{eq:QP2} separates and the fact that  $\bU^*$ and $\bV^*$ are independent, the (marginal) fiducial distribution for the failure distribution function $F$ is
\begin{equation}\label{eq:GFDsurv}
Q^F(\by,\bdelta,\bU^*) \mid \{Q^F(\by,\bdelta,\bU^*) \neq \emptyset\}.
\end{equation}
The conditional distribution of  $\bU^* \mid \{Q(\by,\bdelta,\bU^*)\neq\emptyset\}$ can be sampled efficiently because it is the distribution of a particular random reordering  of a sample of independent and identically distributed $U(0,1)$.
To this end we define $\mathcal P$ as the set of all permutations for which the permuted order statistics $\bu_{(\Pi)}, \Pi\in\mathcal P$  satisfy $Q^F(\by,\bdelta, \bu_{(\Pi)})\neq\emptyset$. Notice that the $i$-th element of $\bu_\Pi$ is the $\Pi(i)$-th order statistics of $\bu$, i.e., ${\bu_{(\Pi)}}_i=\bu_{(\Pi(i))}$.  The set $\mathcal{P}$ is invariant to $\bu$ as long as $\bu$ has no ties.  Therefore we simulate independent and identically distributed $U(0,1)$, sort them, and then permute them using a permutation selected at random from $\mathcal P$.

The random permutation $\Pi\in\mathcal P$ can be generated sequentially starting from the smallest among the $\by$ to the largest. We start with the set $\mathcal N=\{1,\ldots, n\}$. At any given observation $y_i$, we select $\Pi(i)$ from $\mathcal N$ as either a) the smallest remaining value if the observed value $y_i$ is a failure time or b) any of the remaining values selected at random if the observed value $y_i$ is a censoring time.
We then remove the selected $\Pi(i)$ from $\mathcal N$  and proceed to the next smallest observation $y_{j}$ until we exhaust the observations and $\mathcal N$. 

 Given $\{Q^F(\by,\bdelta,\bU^*) \neq \emptyset\}$, and the results of the first $i-1$ steps, the components of $\bU^{\ast}$ not yet selected are exchangeable, which validates the proposed algorithm.

The details of this algorithm are in the Appendix. 
We implement the same two basic approaches to deriving statistical procedures from the GFD as in  Section~\ref{nonsur}.
To illustrate the fiducial distribution in the right censoring case, failure time $X$ follows  $Weibull(20,10)$ and censoring time $Z$ follows $Exp(20)$ with sample size 300. Censoring percentage is about 60\%. We plot a fiducial sample of the survival function $S^I_i (i=1,\ldots,1000)$ and Kaplan-Meier estimator in the right panel of Figure~\ref{both}. As expected, we see a wider spread of fiducial curves in the censoring case indicating higher uncertainty.
\begin{figure}[t]\centering
\includegraphics[width=60mm]{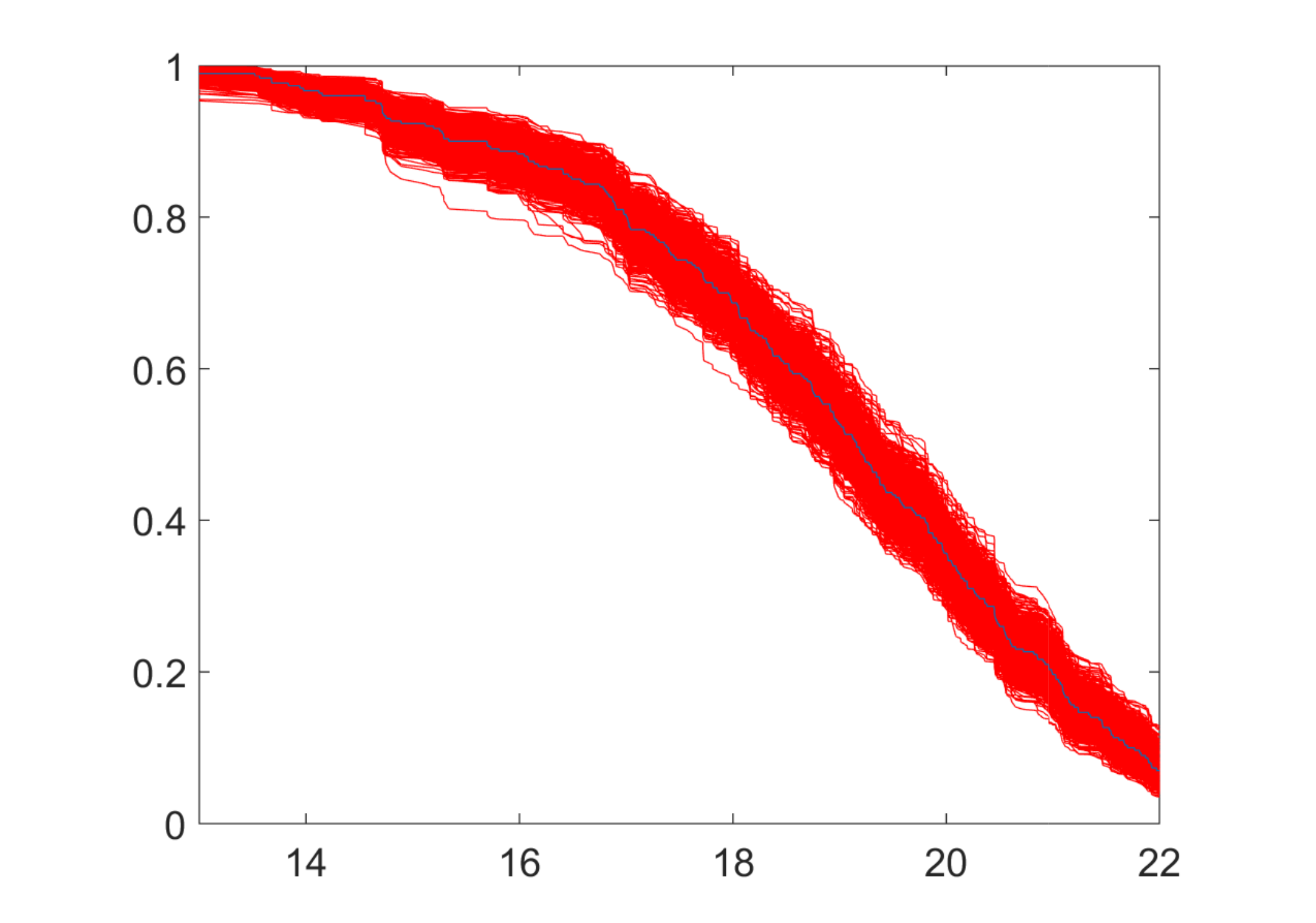}
\includegraphics[width=60mm]{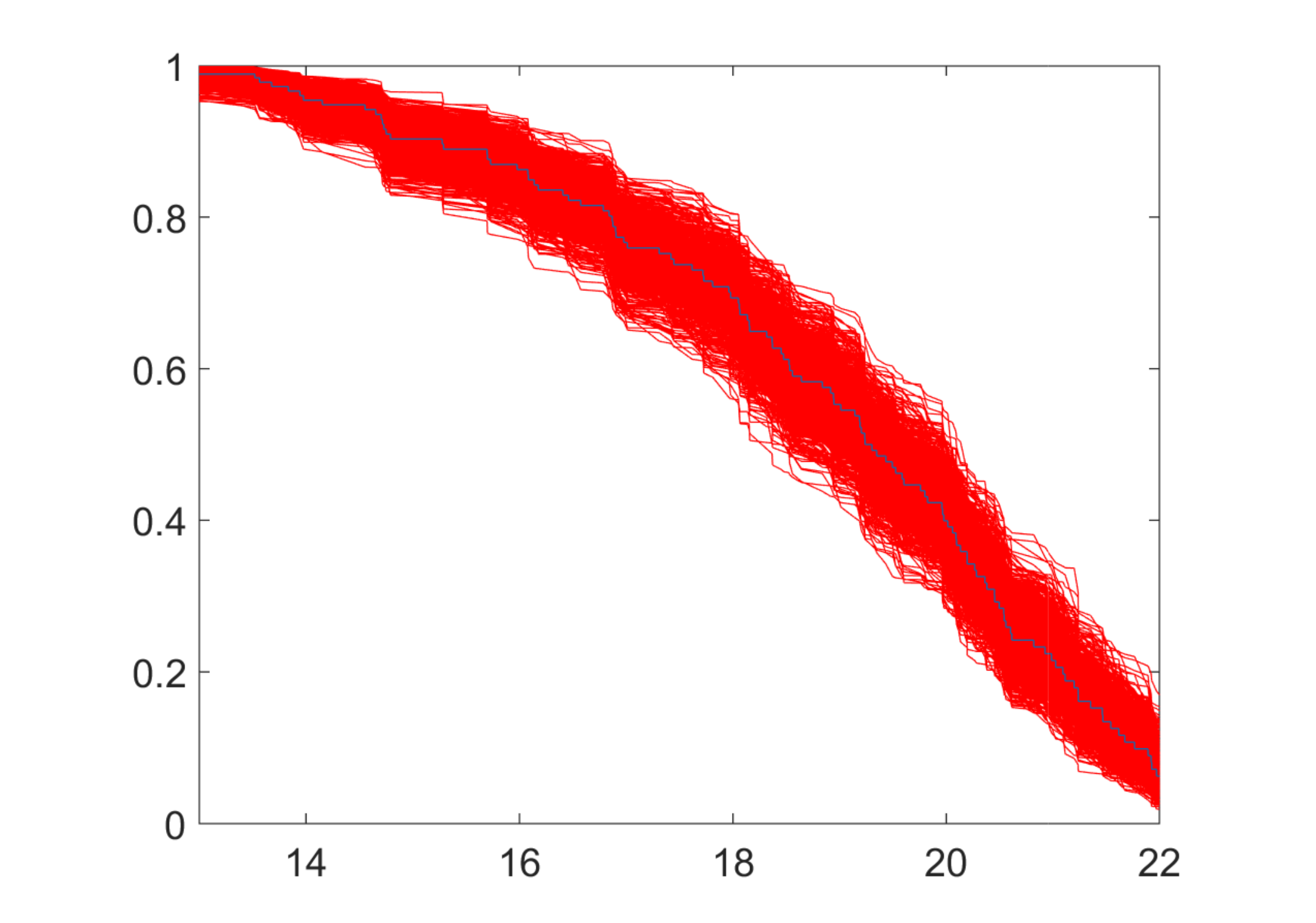}
\caption{ A plot of Monte Carlo realizations $S^I_i (i=1,\ldots,1000)$ sampled from the GFD based on a sample of 300 uncensored $Weibull(20,10)$ observations, 
and the same 300 $Weibull(20,10)$ observations censored by $Exp(20)$. The red curves are  the 1000 fiducial curves, and the blue curve are the empirical survival function and the Kaplan-Meier estimator, respectively. As expected, we observe higher uncertainty in the fiducial sample under censoring.}
\label{both}
\end{figure}

\subsection{Inference based on fiducial distribution}\label{sec3}

In this section, we describe how to use GFD for inference, specifically, point estimation, pointwise confidence intervals for survival functions and quantiles, curvewise confidence intervals, and testing. The actual numerical implementation will be based on a fiducial sample of survival functions $S_i^{L},S_i^{U}$, and $S_i^{I}~(i=1,\ldots,m)$, i.e., the lower bound, the upper bound, and the log-linear interpolation respectively, obtained from the algorithm for generating Monte Carlo samples from the GFD described in the Appendix.

By Lemma \ref{a:lemma2} shown in the Appendix, the Kaplan-Meier estimator falls into the interval given by the expectation of the lower and upper fiducial bounds at any failure time $t$. However, instead of using 
 the Kaplan-Meier estimator we propose to use the pointwise median of the log-linear interpolation fiducial distribution as a point estimator of the survival function. It follows from Section~\ref{TR} that the proposed estimator is asymptotically equivalent to the Kaplan-Meier estimator.  Numerically, we estimate the median of the GFD at time $x$ by computing a pointwise median of the fiducial sample $S_i^{I}(x)~(i=1,\ldots,m)$.  We report a simulation study in Section~\ref{CIsimu} to support this estimator.

As explained at the end of Section~\ref{nonsur} we use two types of pointwise confidence intervals, conservative and log-linear interpolation, using quantiles of appropriate parts of the fiducial samples. For example, a $95\%$ confidence log-linear interpolation confidence interval for $S(x)$ is formed by using the empirical 0$\cdot$025 and 0$\cdot$975 quantiles of $S_i^{I}(x)$. Similarly, a $95\%$ conservative confidence interval is formed by taking the empirical 0$\cdot$025 quantile of $S_i^{L}(x)$ as a lower limit and the empirical 0$\cdot$975 quantile of $S_i^{U}(x)$ as an upper limit.
Simulation results in Section~\ref{CIsimu} show that the proposed confidence intervals match or outperform their main competitors regarding coverage and length.

In order to save space, in the rest of this section we present procedures based on the log-linear interpolation sample only. A conservative version can be obtained analogously. In survival analysis, we are also interested in confidence intervals for quantile $q$ of the survival function, where $0<q<1$. We obtain such a confidence interval by inverting the procedure of computing the pointwise confidence interval. Specifically, a 95\% confidence interval is obtained by taking empirical 0$\cdot$025 and 0$\cdot$975 quantiles of the inverse of fiducial sample $S_i^{I}$ evaluated at $q$.

%

Next, we discuss the use of the GFD to obtain simultaneous curvewise confidence bands. In particular, for a $1-\alpha$ curvewise confidence set we propose using a band $\{S:\|S-M\|\leq c\}$ of fiducial probability $1-\alpha$, where $M$ denotes the pointwise median of the GFD, and $\|\cdot\|$ is the $L_{\infty}$ norm, i.e., $\|S-M\|=\max\limits_{x}|S(x)-M(x)|$. Numerically we implement this by using a fiducial sample. Let
\[
l_j=\|S^I_j-\hat M\|=\max\limits_{x}|S^I_j(x)-\hat M(x)|, j=1,\ldots, m,
\]
 where $\hat M$ is the estimated pointwise median of the GFD. Then we form the 95\% curvewise confidence band $\{S:\|S-\hat M\|\leq \hat c\}$, where $\hat c$ is the 0$\cdot$95 quantile of $l_j$. To illustrate, we plot 95\% pointwise and curvewise confidence intervals for the $Weibull(20,10)$ example under right censoring in Figure~\ref{ci}.

\begin{figure}[t]\centering
\includegraphics[width=60mm]{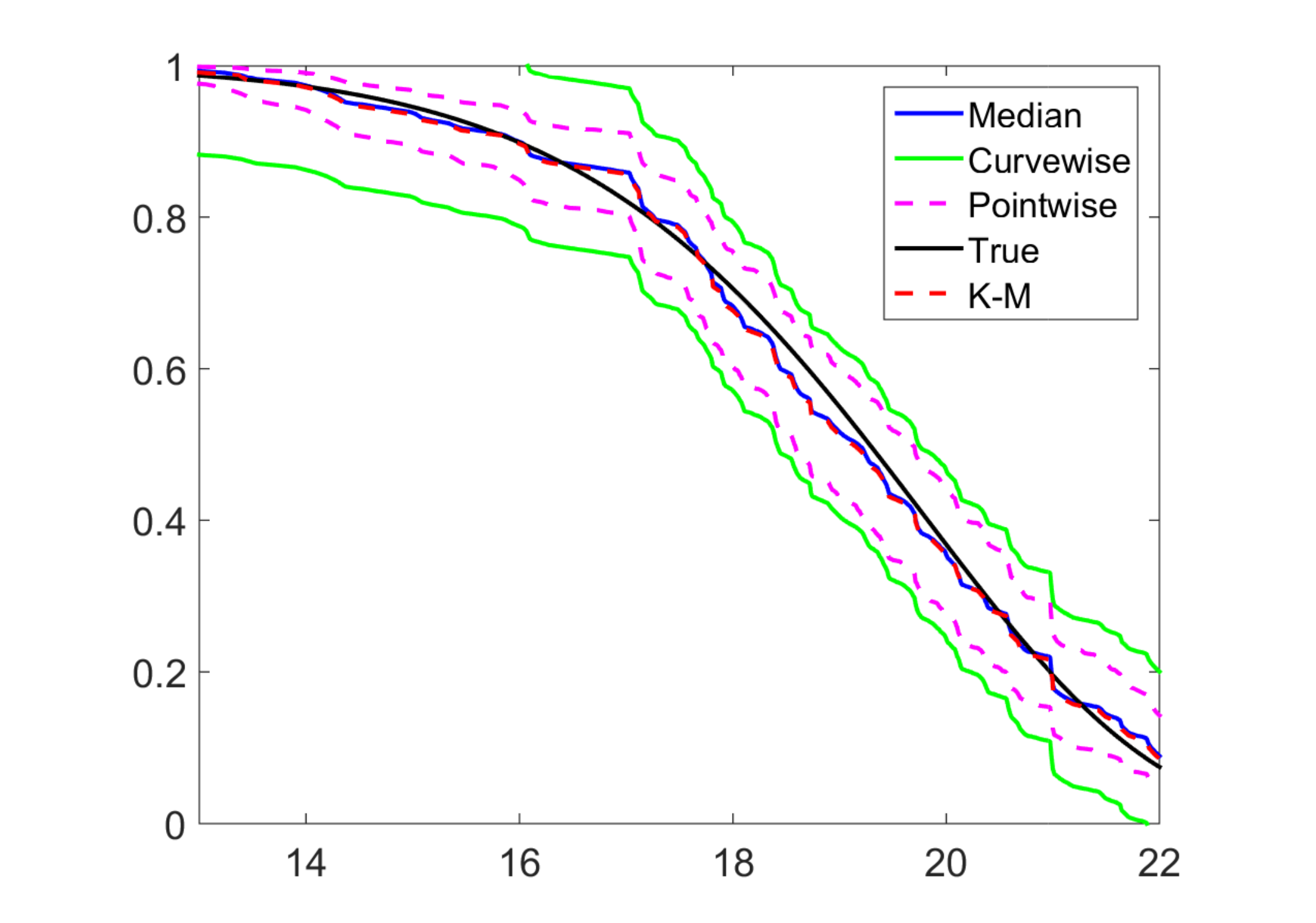}
\caption{An example of 95\% pointwise and curvewise confidence intervals of survival function by proposed log-linear interpolation approach.} 
\label{ci}
\end{figure}

The curvewise confidence set could be inverted for testing. The resulting test is different from the log-rank test \citep{mantel1966evaluation} and its modifications.
Based on our definition of  the $1-\alpha$ fiducial band, the fiducial p-value for the two sided test
\[
H_0: S(t)= S_0(t)~ \text{for all $t$},\quad H_1: S(t)\neq S_0(t)~ \text{for some $t$},
\]
 is $\text{pr}^*_{\by,\bdelta}(\|S^I-M\|\geq\|S_0-M\|)$, where $\text{pr}^*_{\by,\bdelta}$ stands for a fiducial probability computed for observed data $(\by, \bdelta)$, $S^I$ stands for a random survival function following the log-linear interpolation GFD, and as before $M$ is the pointwise median of the fiducial distribution. We estimate this p-value from a fiducial sample by finding the largest $\alpha$ for which $1-\alpha$ curvewise confidence set contains $S_0$. In particular, let
\begin{equation}\label{pvalue}
l_0=\max\limits_{x}|S_0(x)-\hat M(x)|, \quad l_j=\max\limits_{x}|S^I_j(x)-\hat M(x)|,\ j=1,\ldots,m.
\end{equation}
Numerically, we approximate the p-value by the proportion of  the fiducial sample satisfying $l_j \geq l_0 $. 

While the log-rank test is a two sided test only, the fiducial approach could also be used to define one sided tests.  For example for testing
\[
H_0: S(t)\geq S_0(t) ~ \text{for all $t$},\quad H_1: S(t)<S_0(t) ~ \text{for some $t$},
\]
we define a fiducial p-value as the fiducial probability $\text{pr}^*_{\by,\bdelta}(\max\limits_{x} \{S^I(x)-M(x)\} \geq \max\limits_{x}\{S_0(x)-M(x)\})$. 

Finally, let us consider two sample testing. For each sample, we have observed values $\by^i$ and censoring indicators $\bdelta^i$, $i=1,2$. The two independent log-linear interpolation GFDs are denoted by $S^{I}_{(\by_i,\bdelta_i)}$, $i=1,2$. When testing  $H_0:S^1-S^2=\Delta_0$ we define a fiducial p-value as the fiducial probability $\text{pr}^*_{\by,\bdelta}(\|S^{I}_{(\by_1,\bdelta_1)}-S^{I}_{(\by_2,\bdelta_2)}-M_D\|\geq \|\Delta_0-M_D\|)$, where $M_D$ is the median of the difference of the two GFDs.

Numerically, we evaluate the p-value in the same fashion as in Equation~\eqref{pvalue}. We will compare the performance of the proposed fiducial test with the log-rank test and sup log-rank test with different weights for the two sample settings by simulation in Section \ref{simulationcom2}. 

\section{Theoretical results}\label{TR}

Recall that the GFD is a data dependent distribution
$\text{pr}^*_{\by,\bdelta}$ that is defined for every fixed data set $(\by,\bdelta)$. It can be made into a random measure $\text{pr}^*_{\bY,\bdelta}$ in the same way as one defines the usual conditional distribution, i.e., by plugging random variables $(\bY,\bdelta)$ for the observed data set. In this section, we will study the asymptotic behavior of this random measure assuming there are no ties with probability 1.

\cite{praestgaard1993exchangeably} prove a Bernstein-von Mises theorem for the exchangeably weighted bootstrap, of which the Bayesian bootstrap \citep{rubin1981bayesian} is an example. 
However, the result of \cite{praestgaard1993exchangeably} is not applicable in the survival settings due to the fact that the jump sizes of $F^L$ or $F^U$ are not exchangeable. In this section, we study the theoretical properties of the GFD in the survival setting. For simplicity, we state the results in this section using upper fiducial bound of survival functions $S^U$, i.e., the lower fiducial bound of cumulative distribution functions $F^L$. Lemma~\ref{a:lemma1} in the Appendix proves that the same results hold for $S^L$ and $S^I$.

First we introduce some notations: $X_i$ is failure time, $Z_i$ is censoring time, $Y_i$ is the observed minimum of failure and censoring time, and $\delta_i=I\{X_i \leq Z_i\}$ is the censoring indicator. We define the counting process
\[
N_i(t)=I\{Y_i\leq t\}\delta_i,\quad
 \bar{N}(t)=\sum_{i=1}^n N_i(t),
\]
and the at-risk process
\[
K_i(t)=I\{Y_i \geq t\},\quad
 \bar{K}(t)=\sum_{i=1}^n K_i(t).
\]

We need the following two assumptions which are also needed for theoretical study of the Kaplan-Meier estimator \citep{fleming2011counting}.
\begin{assumption}\label{ass1} There exists a function $\pi$ such that, as $n \rightarrow \infty$,
\begin{align*}
\sup_{0 \leq t < \infty} \left|\bar{K}(t)/n-\pi(t)\right| \rightarrow 0 ~~\text{almost surely}.
\end{align*}
\end{assumption}
This assumption is very mild. For example if $Y_i$ are independent and identically distributed, it is implied by Glivenko-Cantelli Theorem; see the discussion following Assumption~6.2.1 in \cite{fleming2011counting} for more details.

\begin{assumption}\label{ass2} $F_0$ is absolutely continuous.
\end{assumption}

Let $\tilde S(t)=\prod_{s \leq t}\{1-\Delta \bar{N}(s)/\bar{K}(s)\}$ be the Kaplan-Meier estimator. It is well-known,  see for example Theorem 6.3.1 of \cite{fleming2011counting}, that for any $t$ satisfying $\pi(t)>0$,
\begin{equation}\label{KMT}
\sqrt{n} \{\tilde F(\cdot)-F_0(\cdot)\} \rightarrow \{1-F_0(\cdot)\}W\{\gamma(\cdot)\}\quad
\mbox{in distribution on $D[0,t]$,}
\end{equation}
where $\tilde F(t)=1-\tilde S(t)$, $\gamma(t)=\int_0^t\pi^{-1}(s)d\Lambda (s)$, $W$ is Brownian Motion, and $\Lambda$ is the cumulative hazard function.

Recall that the procedure for sampling from \eqref{eq:GFDsurv} in Section~\ref{sur} defines a random permutation $\Pi$. Conditional on  $\{Q^F(\by,\bdelta,\bU^*) \neq \emptyset\}$ and the results of the first $i-1$ steps, the distribution of the $\Pi(i)$-th order statistic $\bU_{(\Pi(i))}^*$ corresponding to a failure time $y_i$ is the minimum of $\bar K(y_i)$ independent random variables distributed as uniform on $(\bU_{(\Pi(j))}^*,1)$, where $\bU_{(\Pi(j))}^*$ corresponds to the failure time $y_j$ immediately preceding $y_i$. If $y_i$ is the smallest failure time  
then set $\bU_{(\Pi(j))}^*=0$.
Since $S^U(y_i)=1-\bU_{(\Pi(i))}^*$ for all failure times, the upper bound of the GFD has a distribution that can be written as
\begin{equation}\label{pest}
S^U(t)=\prod_{s_i \leq t}\{1-\Delta \bar{N}(s_i) B_i\},
\end{equation}
where $\Delta \bar{N}(t)=\bar N(t)-\bar N(t-)$, $s_i$ are ordered failure times, and $B_i$ are independent $Beta(1,\bar K(s_i))$, respectively. Its expectation $\hat{S}(t)=E\{S^U(t)\}$ can be easily computed from  \eqref{pest} as
\begin{equation}\label{KM}
\hat{S}(t)=\prod_{s\leq t}\left\{1-\frac{\Delta \bar{N}(s)}{1+\bar{K}(s)}\right\}.
\end{equation}
Equation \eqref{KM} provides us with a modification of the Kaplan-Meier estimator
that also satisfies \eqref{KMT}. We will use this modification throughout this section and
in all the proofs that can be found in the 
Appendix. As our first result, we prove a concentration inequality for $S^U(t)$.
\begin{theorem}\label{consistency}
The following bound holds for any dataset with $\bar K(t)\geq 1$ and any $\epsilon>0$,
{\small
\begin{equation}\label{eq:concentration}
\text{pr}^*_{\by,\bdelta} \{\sup_{s\leq t}{|S^U(s)-\hat{S}(s)|} \geq 3\epsilon^2/n^{1/2}+\bar{N}(t)/\bar{K}(t)^{-2} \} \leq \bar{N}(t) [ (1-\epsilon/n^{3/4})^{\bar K(t)}+\text{$0$$\cdot$$4$}^{\bar K(t)}+ n/\{\epsilon^2 \bar K(t)\}^2].
\end{equation}
}
\end{theorem}
\begin{remark}\label{remark1}
Theorem~\ref{consistency}  and Assumption~\ref{ass1} imply that the fiducial distribution is uniformly consistent. In particular, provided that we have a sequence of data so that $\bar K(t)/n\to\pi(t)>0$, the right-hand side of \eqref{eq:concentration} is $O(n^{-1})$ whenever $\epsilon^2=n^{1/2}$.
\end{remark}

Before presenting our main result we need two additional assumptions.
\begin{assumption}\label{ass3} 
$\int_0^t f_n(s)/\bar{K}(s) d\bar{N}(s) \rightarrow \int_0^t f(s)\lambda(s)ds$ almost surely for any $t \in \mathcal{I}=\{t:\pi(t)>0\}$ and $f_n \rightarrow f$ uniformly.
\end{assumption}
Assumption \ref{ass3} is reasonable since the probability of failure and censoring both happening in the $[t,t+\Delta t)$ is of a higher order $O((\Delta t)^2)$.

\begin{assumption}\label{ass4}
$\sup_{0 \leq s \leq t} | \tilde F(s)- F_0(s) | \to 0$ almost surely for any $t \in \mathcal{I}=\{t:\pi(t)>0\}$, where $\tilde F=1-\tilde S$, and $\tilde S$ is the Kaplan-Meier estimator.
\end{assumption}
\begin{remark}
 The strong consistency result of Assumption \ref{ass4} has been proved for the model described in Section \ref{sur} by \cite{GCKM,10.2307/2242210}. Moreover, Assumption \ref{ass4} is only needed for establishing a strong version of Theorem \ref{main}, i.e., convergence in distribution almost surely. If the Kaplan-Meier estimator only converges in probability, then the convergence mode in Theorem \ref{main} is in distribution in probability.
\end{remark}

The following theorem establishes a Bernstein-von Mises theorem for the fiducial distribution. In particular, we will show that the fiducial distribution of $n^{1/2} \{F^L(\cdot)-\hat F(\cdot)\}$, where $\hat F(\cdot)=1- \hat S(\cdot)$ and $F^L(\cdot)=1- S^U(\cdot)$, converges in distribution almost surely to the same Gaussian process as in \eqref{KMT}. To understand the somewhat unusual mode of convergence used here, notice that there are two sources of randomness present. One is from the fiducial distribution itself that is derived from each fixed data set. The other is the usual randomness of the data. The mode of convergence here is in distribution almost surely, i.e., the centered and scaled fiducial distribution viewed as a random probability measure on $D[0,t]$ converges almost surely to the Gaussian process described in the right-hand side of Equation \eqref{KMT} using the weak topology on the space of probability measures. 

\begin{theorem}\label{main}  Based on Assumptions 
\ref{ass1}--\ref{ass4}, for any $t \in \mathcal{I}=\{t:\pi(t)>0\}$, $n^{1/2}\{ F^L(\cdot)-\hat F(\cdot)  \} \rightarrow \{1-F_0(\cdot)\}W\{\gamma(\cdot)\}$ in distribution on $D[0,t]$ almost surely, where $\gamma(t)=\int_0^t\pi^{-1}(s)d\Lambda (s)$.
\end{theorem}

Notice that Theorem~\ref{main} implies that the pointwise fiducial confidence intervals are equivalent to the
asymptotic confidence intervals based on the Kaplan-Meyer estimator. This fact can be also seen from Theorem~2 of \cite{fay2013pointwise}.
This is in line with our experience with GFD in parametric settings, i.e., the fiducial procedures are asymptotically as efficient as maximum likelihood.
The following corollary shows that Theorem~\ref{main} also implies that all the pointwise and curvewise confidence intervals described in Section~\ref{sec3} have asymptotically correct coverage. Consequently, the tests described in Section~\ref{sec3} also have asymptotically correct type I error.

\begin{corollary} \label{corollary}
Let $\Psi\{\phi(\cdot)\}$ be a map: $D[0,t]\rightarrow \mathbbm{R}$ satisfying, there exists a function $\psi$ so that
\begin{equation}\label{p1}
\Psi\{\phi(\cdot)\}=\Psi\{-\phi(\cdot)\},\quad
\Psi \{ a\phi(\cdot)\}=\psi(a)\Psi\{\phi(\cdot)\},
\end{equation}
 for all $\phi \in D[0,t]$, $a>0$, the distribution of the random variable $\Psi[\{1-F_0(\cdot)\}W\{\gamma(\cdot)\}]$ is continuous and the $(1-\alpha)$-th quantile of this distribution is unique.

Then, under the assumptions in Theorem \ref{main}, any set $C_{n,\alpha}=\{F:\Psi\{F(\cdot)-\hat F(\cdot)\}\leq \epsilon_{n,\alpha}\}$ with $\text{pr}^*_{\by,\bdelta}(C_{n,\alpha})=1-\alpha$ is a $1-\alpha$ asymptotic confidence set for $F_0$.
\end{corollary}

\section{Simulation study}\label{s:Simul}

\subsection{Coverage of pointwise confidence intervals and mean square error of point estimators}\label{CIsimu}
We present comparisons of frequentist properties of the proposed fiducial confidence intervals with a number of competing methods. We will consider two basic groups of settings, one with heavy censoring from  \cite{fay2013pointwise} and another with a moderate level of censoring from \cite{barber1999symmetric}. In both cases the proposed GFD intervals perform comparable to or better than the reported methods.

First we reproduce the settings in \cite{fay2013pointwise} that have a very high level of censoring. \cite{fay2013pointwise} compared their proposed beta product confidence procedure methods with a number of asymptotic methods. These include  Greenwood by logarithm transformation, the confidence interval on the Kaplan-Meier estimator using Greenwood's variance by logarithm transformation \citep{survival-package}; Modified Greenwood by logarithm transformation which modifies the estimator of variance for the lower limit by multipling the Greenwood's variance estimator by $K(y_i)/K(t)$ at $t$, where $y_i$ is the largest observed survival less than or equal to $t$ \citep{survival-package}; Borkowf by logarithm transformation, which gives wider intervals with more censoring and assumes normality on $\log(\tilde S(t))$, where $\tilde S(t)$ is the Kaplan-Meier estimator \citep{borkowf2005simple}; shrinkage Borkowf by logarithm transformation, which uses a shrinkage estimator of the Kaplan-Meier estimator with a hybrid variance estimator \citep{borkowf2005simple}; Strawderman-Wells, that uses the Edgeworth expansion
for the distribution of the studentized Nelson-Aalen estimator \citep{strawderman1997accurate,strawderman1997accurateb}; Thomas-Grunkemeier, a likelihood ratio method which depends on a constrained product-limit estimator of the survival function \citep{thomas1975confidence}; Constrained Beta, which refers the distribution of $\tilde S(t)$ to a beta distribution subject to some constraints \citep{barber1999symmetric}; nonparametric Bootstrap
\citep{efron1981censored,akritas1986bootstrapping}; Constrained~Bootstrap, an improved bootstrap approximation subject to some constraints \citep{barber1999symmetric}.

Simulation studies in \cite{fay2013pointwise} show that the above asymptotic methods have a coverage problem, i.e., the error rate of 95\% confidence interval of all these methods is larger than 5\% in their high censoring scenarios.
Therefore in this setting we focus on comparing the fiducial methods with our main competing methods, which are beta product confidence procedure \citep{fay2013pointwise}, mid-p beta product confidence procedure \citep{fay2016finite}, see also Chapter 11 of \cite{schweder2016confidence}, and Binomial-C \citep{clopper1934use}, which maintain the coverage. We report the error rate of coverage and the average width of confidence intervals for fiducial methods, beta product confidence procedure using method of moment, beta product confidence procedure using Monte Carlo with samples 1000, mid-p beta product confidence procedure, and Binomial-C. We point out that Clopper-Pearson Binomial-C requires knowledge of the censoring times for each individual \citep{fay2013pointwise}.

We consider following two scenarios in \cite{fay2013pointwise}. In the first scenario, failure time $X$ is $Exp(10)$, censoring time $Z$ is $U(0,5)$. 
We simulate 100000 independent datasets of size 30 and applied our methods with fiducial sample size 1000. In the second scenario, 
we reproduce the setting using a mixture of exponentials to mimic the pilot study of treatment in severe systemic sclerosis \citep{nash2007high}. In particular, failure time $X$ is a mixture of $Exp$(0$\cdot$227) with probability 0.187 and $Exp$(22$\cdot$44) with probability 0$\cdot$813, censoring time $Z$ is $U(2,8)$. 
We simulate 100000 independent datasets of size 34 and apply our methods with fiducial sample size 1000.

The simulation results are in Table \ref{table1} and Table \ref{table2} for each scenario, respectively. In the tables, L denotes the error rate that the true parameter is less than the lower confidence limit; U denotes the error rate that the true parameter is greater than the upper confidence limit. The two-sided error rate is obtained by adding the values in column L and U. Values less than 2$\cdot$5\% in individual columns, 5\% in aggregate, indicate good performance. W is the average width of the confidence interval. The row labels are: FD-I the proposed method using log-linear interpolation; FD-C the proposed conservative confidence interval; 
 BPCP-MM beta product confidence procedure using method of moment; BPCP-MC beta product confidence procedure using Monte Carlo; BPCP-MP mid-p beta product confidence procedure; BN Clopper-Pearson Binomial-C. From Table \ref{table1} and Table \ref{table2} we see that our confidence intervals using log-linear interpolation maintain the aggregate coverage, are much shorter, but may be slightly biased to the left. Not surprisingly, the performance of the proposed conservative confidence interval is similar to the beta product confidence procedure method. 
  Recall, Table 1 and Table F$\cdot$2 in \cite{fay2013pointwise} show all asymptotic methods mentioned above have a coverage problem in this heavily censored setup, and so are not considered here.

We also perform a simulation for the mean square error of survival functions, adopting a setting in \cite{fay2013pointwise}. Here, failure time is $Exp(1)$, and censoring time is $U(0,5)$. 
We simulate 100000 independent datasets of size 25 and apply our fiducial methods with fiducial sample size 10000. Since the Kaplan-Meier estimator is not defined after the largest observation if it is censored, we follow  \cite{fay2013pointwise} and define it in three ways after the last observation: KML is defined as 0, KMH is defined as the Kaplan-Meier at the last value, and KMM=0$\cdot$5*KML+0$\cdot$5*KMH. We evaluate mean square error at $t$, where $S(t)=$ 0$\cdot$99, 0$\cdot$9, 0$\cdot$75, 0$\cdot$5, 0$\cdot$25, 0$\cdot$1, 0$\cdot$01. We report the results in Table \ref{table3}. FD-I uses the pointwise median of the log-linear interpolation fiducial distribution as a point estimator of the survival function. BPCP-MM and BPCP-MP are associated median unbiased estimators defined in \cite{fay2013pointwise}. We see the proposed fiducial approach has the smallest mean square error for $S(t)=$ 0$\cdot$99, 0$\cdot$9, 0$\cdot$75, 0$\cdot$5, 0$\cdot$25, 0$\cdot$1, 0$\cdot$01.

\begin{table}[H]
\centering
\caption{Error rate (in percent) and average width of $95\%$ confidence intervals for scenario 1}
\label{table1}
\begin{tabular}{ccccccccccccc}
\hline
         & \multicolumn{3}{c}{t=1} & \multicolumn{3}{c}{t=2} & \multicolumn{3}{c}{t=3} & \multicolumn{3}{c}{t=4} \\
         & L      & U     & W      & L      & U     & W      & L      & U     & W      & L      & U     & W      \\
\hline
FD-I   & 1$\cdot$9    & 2$\cdot$7   & 0$\cdot$21   & 1$\cdot$5    & 2$\cdot$8   & 0$\cdot$29   & 1$\cdot$4    & 3$\cdot$0   & 0$\cdot$37   & 1$\cdot$8    & 3$\cdot$1   & 0$\cdot$45   \\
FD-C   & 0$\cdot$0    & 1$\cdot$4   & 0$\cdot$26   & 0$\cdot$3    & 1$\cdot$6   & 0$\cdot$36   & 0$\cdot$1    & 1$\cdot$5   & 0$\cdot$46   & 0$\cdot$0    & 1$\cdot$4   & 0$\cdot$63   \\
BPCP-MM & 0$\cdot$0    & 1$\cdot$3   & 0$\cdot$26   & 0$\cdot$3    & 1$\cdot$4   & 0$\cdot$35   & 0$\cdot$1    & 1$\cdot$3   & 0$\cdot$46   & 0$\cdot$0    & 1$\cdot$0   & 0$\cdot$62   \\
BPCP-MC & 0$\cdot$0    & 1$\cdot$3   & 0$\cdot$25   & 0$\cdot$4    & 1$\cdot$5   & 0$\cdot$35   & 0$\cdot$1    & 1$\cdot$5   & 0$\cdot$46   & 0$\cdot$0    & 1$\cdot$4   & 0$\cdot$63   \\
BPCP-MP & 0$\cdot$0    & 2$\cdot$2   & 0$\cdot$23   & 0$\cdot$8    & 2$\cdot$3   & 0$\cdot$32   & 0$\cdot$4    & 2$\cdot$2   & 0$\cdot$41   & 0$\cdot$0    & 2$\cdot$0   & 0$\cdot$57   \\
BN       & 0$\cdot$0    & 1$\cdot$4   & 0$\cdot$26   & 0$\cdot$7    & 1$\cdot$3   & 0$\cdot$38   & 0$\cdot$6    & 1$\cdot$3   & 0$\cdot$51   & 0$\cdot$1    & 0$\cdot$9   & 0$\cdot$70    \\
\hline
\end{tabular}
\end{table}

\begin{table}[H]
\centering
\caption{Error rate (in percent) and average width of $95\%$ confidence intervals for scenario 2}
\label{table2}
\begin{tabular}{ccccccccccccc}
\hline
         & \multicolumn{3}{c}{t=3} & \multicolumn{3}{c}{t=4} & \multicolumn{3}{c}{t=5} & \multicolumn{3}{c}{t=6} \\
         & L      & U     & W      & L      & U     & W      & L      & U     & W      & L      & U     & W      \\
\hline
FD-I   & 2$\cdot$2    & 2$\cdot$7   & 0$\cdot$29   & 1$\cdot$9    & 2$\cdot$9   & 0$\cdot$31   & 1$\cdot$7    & 3$\cdot$0   & 0$\cdot$33   & 1$\cdot$5    & 3$\cdot$2   & 0$\cdot$36   \\
FD-C   & 1$\cdot$2    & 1$\cdot$7   & 0$\cdot$33   & 0$\cdot$7    & 1$\cdot$8   & 0$\cdot$36   & 0$\cdot$4    & 1$\cdot$8   & 0$\cdot$40   & 0$\cdot$1    & 1$\cdot$7   & 0$\cdot$46   \\
BPCP-MM & 1$\cdot$3    & 1$\cdot$7   & 0$\cdot$33   & 0$\cdot$7    & 1$\cdot$7   & 0$\cdot$35   & 0$\cdot$4    & 1$\cdot$6   & 0$\cdot$39   & 0$\cdot$1    & 1$\cdot$4   & 0$\cdot$46   \\
BPCP-MC & 1$\cdot$2    & 1$\cdot$8   & 0$\cdot$32   & 0$\cdot$7    & 2$\cdot$0   & 0$\cdot$35   & 0$\cdot$4    & 1$\cdot$9   & 0$\cdot$39   & 0$\cdot$1    & 1$\cdot$9   & 0$\cdot$46   \\
BPCP-MP & 1$\cdot$8    & 2$\cdot$1   & 0$\cdot$30   & 1$\cdot$6    & 2$\cdot$4   & 0$\cdot$32   & 0$\cdot$9    & 2$\cdot$5   & 0$\cdot$36   & 0$\cdot$4    & 2$\cdot$3   & 0$\cdot$41  \\
BN       & 1$\cdot$4    & 1$\cdot$5   & 0$\cdot$35   & 1$\cdot$5    & 1$\cdot$6   & 0$\cdot$40   & 1$\cdot$5    & 1$\cdot$7   & 0$\cdot$46   & 1$\cdot$0    & 1$\cdot$5   & 0$\cdot$56    \\
\hline
\end{tabular}
\end{table}

\begin{table}[H]
\centering
\small
\caption{Mean square error of survival function estimators}
\label{table3}
\begin{tabular}{cccccccc}
\hline
         & $S(t)=$0$\cdot$99 & $S(t)=$0$\cdot$9  & $S(t)=$0$\cdot$75 & $S(t)=$0$\cdot$5   & $S(t)=$0$\cdot$25 & $S(t)=$0$\cdot$1  & $S(t)=$0$\cdot$01 \\
\hline
  FD-I & 0$\cdot$30 & 3$\cdot$11 & 7$\cdot$08 & 10$\cdot$08 & 8$\cdot$24 & 4$\cdot$38 & 1$\cdot$20\\
  BPCP-MM & 0$\cdot$44 & 3$\cdot$44 & 7$\cdot$50 & 10$\cdot$60 & 8$\cdot$83 & 4$\cdot$40 & 1$\cdot$50 \\
  BPCP-MP & 0$\cdot$48 & 3$\cdot$65 & 7$\cdot$54 & 10$\cdot$62 & 8$\cdot$99 & 5$\cdot$79 & 0$\cdot$26 \\
  KML & 0$\cdot$39 & 3$\cdot$61 & 7$\cdot$71 & 10$\cdot$94 & 9$\cdot$38 & 6$\cdot$17 & 0$\cdot$28 \\
  KMM & 0$\cdot$39 & 3$\cdot$61 & 7$\cdot$71 & 10$\cdot$94 & 9$\cdot$35 & 5$\cdot$77 & 0$\cdot$79 \\
  KMH & 0$\cdot$39 & 3$\cdot$61 & 7$\cdot$71 & 10$\cdot$94 & 9$\cdot$33 & 5$\cdot$65 & 2$\cdot$92 \\
\hline
\end{tabular}
\end{table}

Our second simulation study setting comes from \cite{barber1999symmetric} where the data contains more exact observations. In the first scenario, survival time $X$ follows $Exp(10)$, and censoring time $Z$ is $Exp(50)$. In the second scenario, survival time $X$ follows $Exp(10)$, and censoring time $Z$ is $Exp(25)$. We plot the empirical error rates from 5000 simulations with sample size $n = 100$ of different non-asymptotic confidence intervals in the Figures~\ref{picture1} and  \ref{picture2}, respectively. From the Figures \ref{picture1}, Figure \ref{picture2}, and the figures in \cite{barber1999symmetric}, we see that the fiducial confidence intervals do as well as the constrained bootstrap in these settings.

\begin{figure}[H]\centering
\includegraphics[width=150mm]{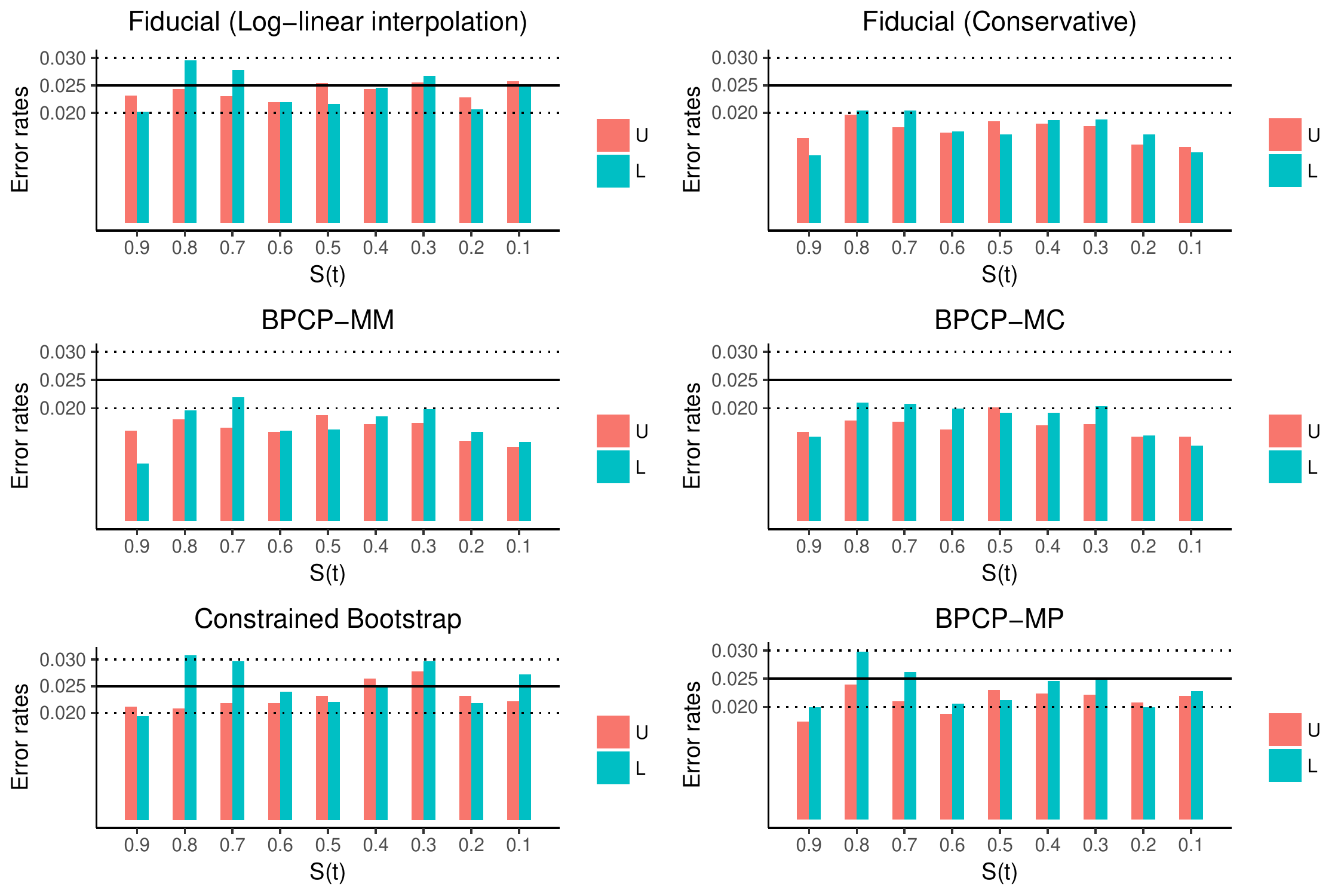}
\caption{Error rate from 5000 simulations of different confidence intervals with $n$ = 100, survival time follows $Exp(10)$, and censoring time follows $Exp(50)$. L denotes the error rate that the true parameter is lower than lower bound. U denotes the error rate that the true parameter is above the upper bound.}
\label{picture1}
\end{figure}

\begin{figure}[H]\centering
\includegraphics[width=150mm]{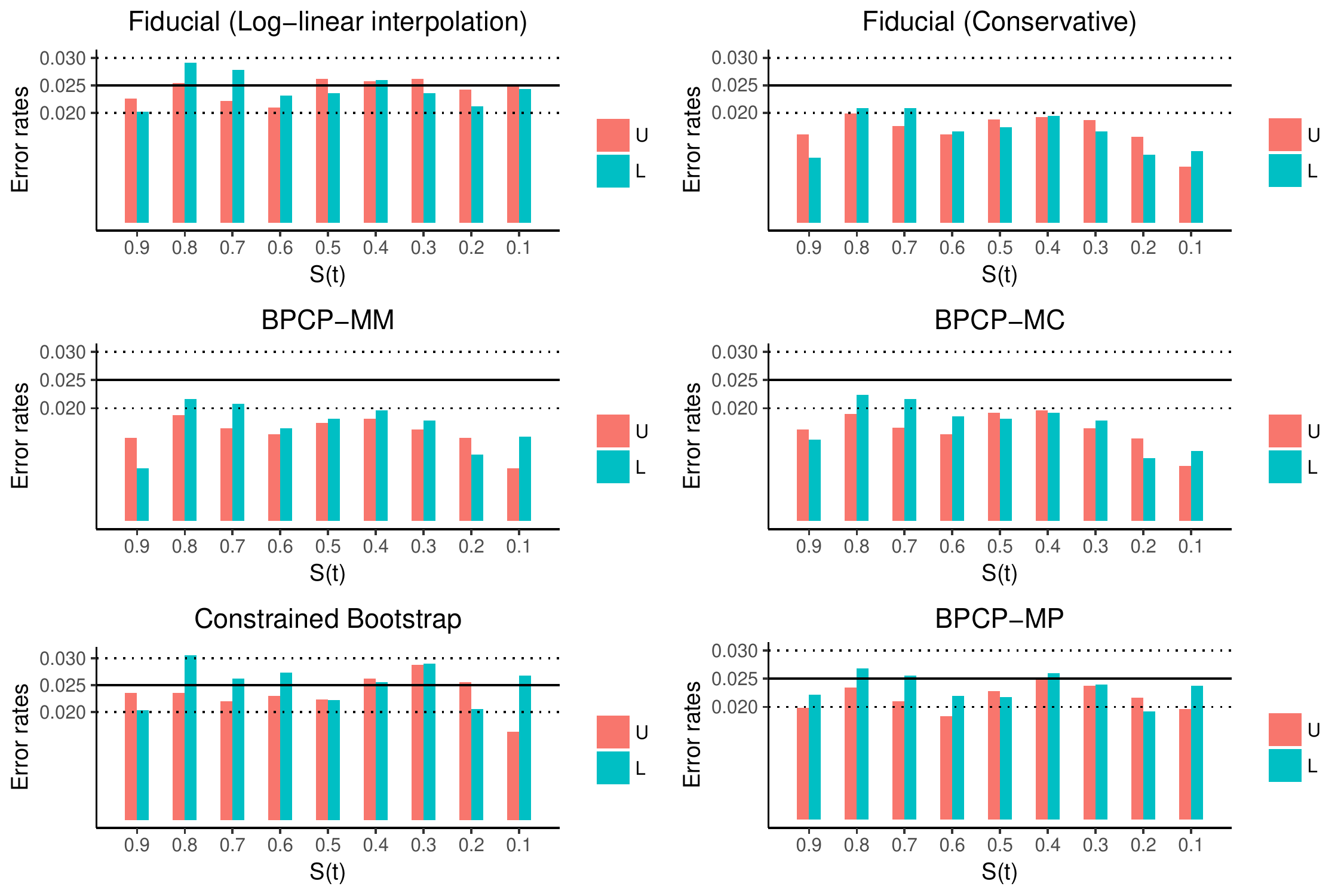}
\caption{Error rate from 5000 simulations of different confidence intervals with $n$ = 100, survival time follows $Exp(10)$, and censoring time follows $Exp(25)$. L denotes the error rate that the true parameter is lower than lower bound. U denotes the error rate that the true parameter is above the upper bound.}
\label{picture2}
\end{figure}

\subsection{Comparisons between the proposed fiducial test and different types of log-rank tests for two sample testing}\label{simulationcom2}
We compare the performance of the proposed fiducial approach with different types of tests for testing the equality of two survival functions \citep{dardis2016survmisc}. A common approach to testing the difference of two survival curves is the log-rank test. There are several modifications of the log-rank tests that consist of re-weighing. In our tables, LR denotes the original log-rank test with weight 1 \citep{mantel1966evaluation};
GW, i.e., Gehan-Breslow generalized Wilcoxon, denotes log-rank test weighted by the number at risk overall \citep{gehan1965generalized}; TW denotes log-rank test weighted by the square root of the number at risk overall \citep{tarone1977distribution}; PP denotes log-rank test with Peto-Peto's modified survival estimate \citep{peto1972asymptotically}; MPP denotes log-rank test with modified Peto-Peto's survival estimate \citep{andersen1982cox}; FH denotes Fleming-Harrington weighted log-rank test \citep{harrington1982class}. The supremum family of tests are designed to detect differences in survival curves which cross \citep{fleming1987supremum,eng2005sample}. SLR denotes the original sup log-rank test with weight 1;
SGW denotes the sup version of GW; STW denotes the sup version of TW; SPP denotes the sup version of PP; SMPP denotes the sup version of MPP; SFH denotes the sup version of FH.

 Four scenarios are considered in this section. In the first scenario the null hypothesis is true. In the remaining three scenarios we consider various departures from the null hypothesis.  For each scenario we simulated 500 independent datasets of size 200, and applied the proposed fiducial test with fiducial sample size 1000 as well as the 12 existing methods mentioned above. Then we calculate the percentage of p-values less than 0$\cdot$05. If the null hypothesis is true, the p-value should follow uniform distribution and the percentage should be around 5\%. If the null hypothesis is false, a higher percentage is preferable as it means bigger power.


In the first scenario, for the first group, failure time is $Weibull(2,1)$ and censoring time follows $|N(0,1)|$. The censoring percentage is approximately 55\%. For the second group, failure time is again $Weibull(2,1)$  but censoring time is $Exp(1)$. The censoring percentage is approximately 60\%. We observe that p-values of all methods follow uniform distribution under $H_0$. Table \ref{table4} shows the percentage of p-value less than 0$\cdot$05. The percentages of p-value less than 0$\cdot$05 of all methods are about 0$\cdot$05.
\begin{table}[!h]
\centering
\small
\caption{Percentage of p-value less than 0$\cdot$05 ($\%$)}
\label{table4}
\begin{tabular}{ccccccccccccc}
\hline
          Fiducial & LR  & GW  & TW   & PP  & MPP  & FH  & SLR & SGW & STW  & SPP & SMPP     & SFH  \\
\hline
 5$\cdot$0 &  5$\cdot$0   &6$\cdot$6   & 6$\cdot$4   & 6$\cdot$4   & 6$\cdot$0   & 4$\cdot$8   & 4$\cdot$6    & 6$\cdot$0   & 6$\cdot$0   &6$\cdot$0    &6$\cdot$0   & 4$\cdot$2    \\
\hline
\end{tabular}
\end{table}

In the second scenario for the first group, failure time follows $Exp(30)$ and censoring time follows $Exp(30)$. The censoring percentage is about 50\%. For the second group, we use $Weibull(30,20)$ to generate failure time,  and $Exp(30)$ for censoring time with censoring percentage of about 50\%. The power of the test at the $\alpha=$0$\cdot$05 level, i.e. the proportion of $p<$ 0$\cdot$05 is shown in Table~\ref{table5}. 
In this scenario, the proposed fiducial test is as powerful as the sup log-rank tests.

\begin{table}[!h]
\centering
\small
\caption{Percentage of p-value less than 0$\cdot$05 ($\%$)}
\label{table5}
\begin{tabular}{ccccccccccccc}
\hline
         Fiducial & LR  & GW  & TW   & PP  & MPP  & FH  & SLR & SGW & STW  & SPP & SMPP     & SFH   \\
\hline
100 & 71$\cdot$0  &98$\cdot$4   & 27$\cdot$6   & 49$\cdot$2   & 50$\cdot$8   & 100   & 100    & 100  & 100   & 100   & 100   & 100    \\
\hline
\end{tabular}
\end{table}


In the third scenario, for the first group, let $Weibull(30,20)$ be the distribution of failure time and $U(0,80)$ be the distribution of censoring time. The censoring percentage is about 25\%. For the second group, let $Weibull(20,20)$ be the distribution of failure time and $U(0,80)$ be the distribution of censoring time. The censoring percentage is about 20\%. The power of the test at the $\alpha=$ 0$\cdot$05 level, i.e. the proportion of $p<$ 0$\cdot$05 is shown in Table~\ref{table6}. We see that only SGW, SPP, SMPP and the proposed fiducial test have power larger than half at $\alpha=$ 0$\cdot$05 level. 
\begin{table}[!h]
\centering
\small
\caption{Percentage of p-value less than 0$\cdot$05 ($\%$)}
\label{table6}
\begin{tabular}{ccccccccccccc}
\hline
      Fiducial &    LR  & GW  & TW   & PP  & MPP  & FH  & SLR & SGW & STW  & SPP & SMPP     & SFH    \\
\hline
  54$\cdot$2 & 21$\cdot$4  &15$\cdot$2   & 4$\cdot$8   & 14$\cdot$0   & 14$\cdot$4   & 39$\cdot$4   & 26$\cdot$6    & 55$\cdot$0  & 39$\cdot$4   & 53$\cdot$8   & 54$\cdot$0   & 29$\cdot$6     \\
\hline
\end{tabular}
\end{table}

In the fourth scenario, for the first group, failure time follows $Exp(1)$, and censoring time follows $|N(0,1)|$ with censoring percentage of about 50\%. For the second group, failure time is $|N(0,1)|$ censored by $Weibull(2,1)$. The censoring percentage is about 40\%. The power of the test at the $\alpha=$ 0$\cdot$05 level, i.e. the proportion of $p<$ 0$\cdot$05 is shown in Table~\ref{table7}. We see that only FH, SFH, and the proposed fiducial test have power larger than 0$\cdot$1 at the $\alpha=$ 0$\cdot$05 level. FH seems to use better weights than other log-rank tests, however, the proposed fiducial test doesn't need to specify any weight and is better than FH in this scenario.

\begin{table}[!h]
\centering
\small
\caption{Percentage of p-value less than 0$\cdot$05 ($\%$)}
\label{table7}
\begin{tabular}{ccccccccccccc}
\hline
    Fiducial  &      LR  & GW  & TW   & PP  & MPP  & FH  & SLR & SGW & STW  & SPP & SMPP     & SFH   \\
\hline
19$\cdot$0 & 7$\cdot$8  &5$\cdot$4   & 4$\cdot$8   & 4$\cdot$6   & 4$\cdot$6   & 16$\cdot$2   & 6$\cdot$6    & 7$\cdot$4  & 5$\cdot$4   & 5$\cdot$4   & 5$\cdot$4   & 10$\cdot$6     \\
\hline
\end{tabular}
\end{table}

\section{Gastric tumor study}\label{realdata}
In this section, we analyze the following dataset presented in \cite{klein2005survival}.
A clinical trial of chemotherapy against chemotherapy combined with radiotherapy in the treatment of locally unresectable gastric cancer was conducted by the Gastrointestinal Tumor Study Group \cite{schein1982comparison}. In this trial, forty-five patients were randomized to each of the two groups and followed for several years. We draw the Kaplan-Meier curves for these two datasets in Figure \ref{fig17}.

\begin{figure}[H]\centering
\subfloat[Kaplan-Meier estimators for two treatment groups.]{{\label{fig17}} \includegraphics[height=48mm]{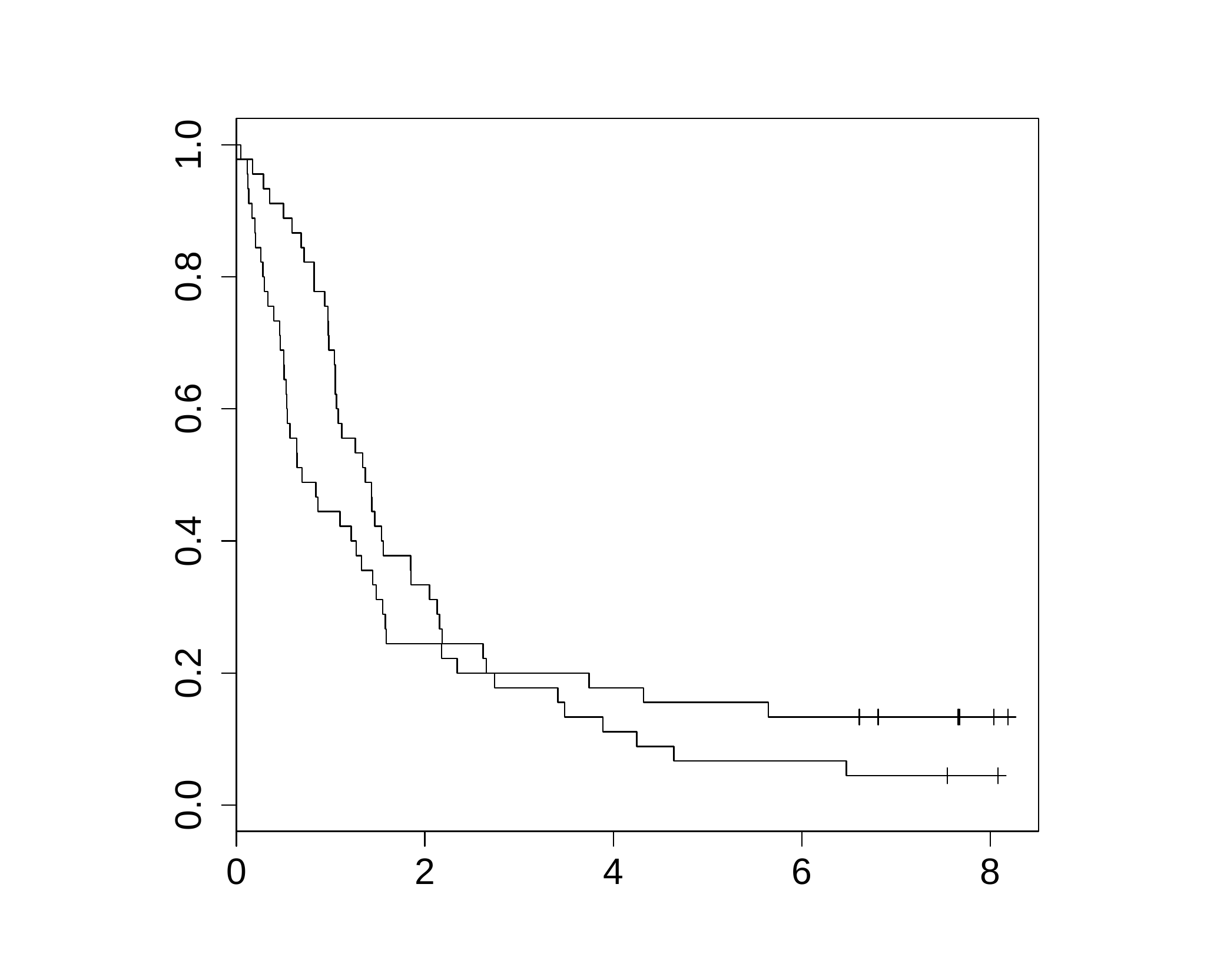}}  \quad
\subfloat[Difference of two sample fiducial distributions using log-linear interpolation.]{{\label{realdiff} } \includegraphics[height=45mm]{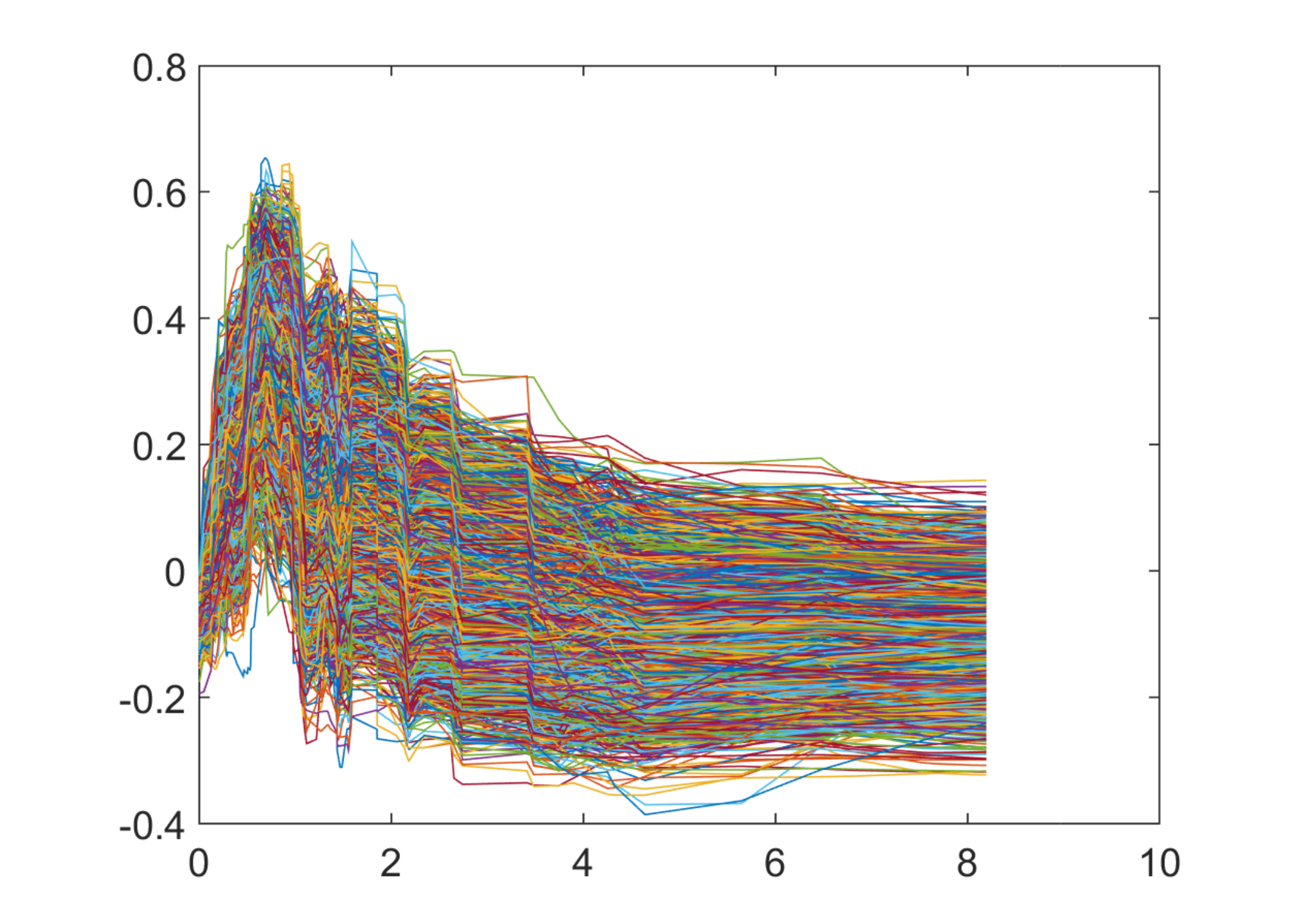}} \\
\caption{}
\end{figure}

By examining the plot in Figure~\ref{fig17} we notice that the two hazards appear to be crossing which could pose a problem for some log-rank tests.  Table~\ref{tablereal1} reports p-values obtained using the same 13 tests described in Section \ref{simulationcom2}.

\begin{table}[!h]
\centering
\small
\caption{p-value of different tests (in \%)}
\label{tablereal1}
\begin{tabular}{cccccccccccccc}
\hline
       & Fiducial  & LR  & GW  & TW   & PP  & MPP  & FH  & SLR & SGW & STW  & SPP & SMPP     & SFH    \\
\hline
p  & 0$\cdot$2 & 63$\cdot$5    & 4$\cdot$6   & 16$\cdot$8   & 4$\cdot$6    & 4$\cdot$3   & 90$\cdot$6   & 5$\cdot$6    & 0$\cdot$6   & 1$\cdot$5   & 0$\cdot$6    & 0$\cdot$6   & 22$\cdot$8     \\
\hline
\end{tabular}
\end{table}

The proposed fiducial test gives the smallest p-value of 0$\cdot$002. To explain why the fiducial approach works on this dataset, we plot the sample of the difference of two fiducial distributions in Figure~\ref{realdiff}. If these two datasets are from the same distribution, 0 should be well within the sample curves. However, from the picture, we could see that the majority of curves are very far away from 0 on the interval $[0,1]$.\\

In order to study the power of our test in this situation, we present a simulation study. We use the data to estimate the failure and censoring distribution for both datasets. Then we use these estimated distributions as truth to generate 500 synthetic datasets that mimic our data. On each dataset, we perform the proposed fiducial test with fiducial sample size 1000 and the 12 different types of log-rank tests. Table~\ref{tablereal2} shows the percentage of p-value less than 0$\cdot$05. We see that the proposed fiducial test has the best power.

\begin{table}[!h]
\centering
\small
\caption{Percentage of p-value less than 0$\cdot$05 ($\%$)}
\label{tablereal2}
\begin{tabular}{ccccccccccccc}
\hline
      Fiducial &     LR  & GW  & TW   & PP  & MPP  & FH  & SLR & SGW & STW  & SPP & SMPP     & SFH   \\
\hline
 87$\cdot$4  & 10$\cdot$2    & 53$\cdot$4   & 31$\cdot$0   & 53$\cdot$0   & 53$\cdot$4   & 7$\cdot$6   & 57$\cdot$6    & 84$\cdot$6   & 78$\cdot$4   & 84$\cdot$6    & 84$\cdot$6   & 23$\cdot$2   \\
\hline
\end{tabular}
\end{table}

\section{Discussion}\label{s:Conclude}

In this paper we derived a nonparametric generalized fiducial distribution for right censored data. This GFD provided us with a unified framework for deriving statistical procedures such as pointwise and curvewise approximate confidence intervals and tests. This is to our knowledge the first time the fiducial distribution has been derived for a non-trivial nonparametric model. We proved a functional Bernstein-von Mises theorem which established the asymptotic correctness of the inference procedures based on our GFD. 
Additionally, our simulation studies suggest that our GFD inference procedures are as good and in some instances better than the many other statistical procedures proposed for the various aspects of this classical problem.
Overall, we view generalized fiducial inference in a similar way as maximum likelihood,  as a general purpose approach that provides good quality answers to many statistical problems. As we can see in the paper, the proposed point estimator of survival function is very similar to Kaplan-Meier estimator. However, the strength of the fiducial approach is in uncertainty quantification when the sample size is small. In particular, we recommend using proposed fiducial confidence intervals and tests in the small sample or heavy censoring cases.

 We conclude by listing some open research problems:
\begin{enumerate}

\item We chose to use the sup-norm in the definition of the curvewise confidence intervals and tests. It could be possible to make the procedure somewhat more powerful by using a different (possibly weighted) norm \citep{nair1984confidence}. Similarly, it might be also possible to use the choice of norm for tuning the GFD tests for use against specific alternatives.

\item The proposed fiducial test seems to be relatively powerful against a broad spectrum of alternatives. It would be interesting to implement it inside other statistical procedures where log-rank tests are recursively used, such as imputed survival random forests and their applications \citep{zhu2012recursively,cui2017someasymptotic,CuZhKo2017}.

\item There seems to be an intriguing connection between GFD and empirical likelihood for semi-parametric models \cite[Chapter 11]{schweder2016confidence}. To investigate this connection should make for a fruitful avenue of future research.

\end{enumerate}

\section*{Acknowledgements}
 We thank Michael P. Fay and Jonathan Williams for helpful conversations and suggestions. We thank the editor, associated editor, and reviewers for many useful comments which led to an improved manuscript. Yifan Cui's research was supported in part by the National Science Foundation under Grant No. 1407732. Jan Hannig's research was supported in part by the National Science Foundation under Grant No. 1512945 and 1633074.


\appendix

\section*{Derivation of generalized fiducial distribution under dependence}
We derive GFD for situations when censoring distribution might depend on the failure time.
In particular, consider the following data generating equation:
\begin{align}\label{eq:depStructuralEq}
Y_i=F^{-1}(U_i)\wedge R_i^{-1}\{V_i \mid F^{-1}(U_i)\} \,,\quad \delta_i=I[F^{-1}(U_i)\leq R_i^{-1}\{V_i \mid F^{-1}(U_i)\}].
\end{align}
Here, $R_i^{-1}(v \mid t)$ is the inverse of the conditional distribution function of the censoring time given failure time $t$ specific to the $i$-th subject. Equation  \eqref{eq:depStructuralEq} allows for any within subject dependence between failure and censoring times. 

The corresponding inverse map for a single observation is: If $\delta_i=1$,
\[
Q^{F,R_i}_1(y_i,u_i,v_i)
=\{F: F(y_i)\geq u_i, F(y_i-\epsilon)< u_i ~\text{for any}~ \epsilon>0\}\times\{R_i:R_i^{-1}(v_i \mid y_i)\geq y_i\}.
\]
If $\delta_i=0$, the inverse map for this datum is
\begin{multline*}
Q^{F,R_i}_0(y_i,u_i,v_i)
=\{F, R_i:\  F(y_i)< u_i,\\ R_i\{y_i \mid F^{-1}(u_i)\}\geq v_i, R_i\{y_i-\epsilon \mid F^{-1}(u_i)\}< v_i ~\text{for any}~ \epsilon>0\}.
\end{multline*}
Unlike in \eqref{eq:QP2}, the inverse
$Q^{F,R}(\by,\bdelta,\bu,\bv)=\bigcap_i Q^{F,R_i}_{\delta_i}(y_i,u_i,v_i)$
does not factorize into a Cartesian product. However, the projection of $Q^{F,R}(\by,\bdelta,\bu,\bv)$ onto the failure time distribution margin remains the same as in \eqref{eq:QP2F}, and $Q^{F,R}(\by,\bdelta,\bu,\bv)\neq\emptyset$ if and only if $Q^F(\by,\bdelta,\bu)\neq\emptyset$. Consequently, the marginal fiducial distribution $Q^F(\by,\bdelta,\bu) \mid  Q^{F,R}(\by,\bdelta,\bu,\bv)\neq\emptyset$  is the same as \eqref{eq:GFDsurv}.

Remarkably, the data generating equation \eqref{eq:depStructuralEq}  leads to the same fiducial distribution for failure times as in the independent case given by \eqref{eq:newStructuralEq}. The difference is that unlike in the fully independent case, \eqref{eq:depStructuralEq} does not provide any useful information about the censoring times and can be viewed as allocating all information in the data to the estimation of failure times.


\section*{Proofs}\label{a:proof}
In this section we collect proofs from Section~\ref{TR}.

\begin{proof}[of Theorem~\ref{consistency}]
For simplicity, in this proof, we denote $\text{pr}^*_{\by,\bdelta}$ as pr.
By the definition of $S^U$ and $\hat{S}$,
\begin{equation}\label{con}
\sup_{s\leq t}{|S^U(s)-\hat{S}(s)|}=\sup_{s \leq t}\left|\prod_{i=1}^{\bar N(s)} (1-B_i)-\prod_{i=1}^{\bar N(s)} (1-\frac{1}{1+\bar K(s_i)})\right|,
\end{equation}
where $B_i \sim Beta(1,\bar{K}(s_i))$, $E(B_i)=\{{1+\bar K(s_i)}\}^{-1}$. 

In order to deal with supremum in Equation~\eqref{con}, we use a coupling idea to get
\begin{equation}\label{sumb}
\text{pr}(\sum_{i=1}^{\bar{N}(t)} B_i^2\leq \frac{\epsilon^2}{n^{1/2}})\geq 1-\bar{N}(t) (1-\frac{\epsilon}{n^{3/4}})^{\bar K(t)}.  
\end{equation}
In particular, define $\tilde{B}_i \sim Beta(1,\bar{K}(t))$ generated by the same uniform random variable as $B_i$, so $\tilde{B}_i\geq B_i$. We have
\begin{multline}\label{a3}
\text{pr}(\max_{1\leq i \leq \bar N(t)} B_i \geq \frac{\epsilon}{n^{3/4}}) \leq 
\bar{N}(t) \bar K(t)\int_0^{1-\frac{\epsilon}{n^{3/4}}} \xi^{\bar K(t)-1}d\xi
=\bar{N}(t) (1-\frac{\epsilon}{n^{3/4}})^{\bar K(t)}. 
\end{multline}
Since $\sum_{i=1}^{\bar{N}(t)} B_i^2 \leq \bar{N}(t) \max_{1\leq i \leq \bar N(t)} B_i^2$, further we have
\begin{align*}
&\text{pr}(\sum_{i=1}^{\bar{N}(t)} B_i^2 \geq \frac{\epsilon^2}{n^{1/2}}) 
\leq \text{pr}(\max_{1\leq i \leq \bar N(t)} B_i \geq \frac{\epsilon}{n^{1/4}\{\bar{N}(t)\}^{1/2}}) \leq  \text{pr}(\max_{1\leq i \leq \bar N(t)} B_i \geq \frac{\epsilon}{n^{3/4}}). 
\end{align*}
So Equation~\eqref{sumb} follows.

In order to bound Equation \eqref{con}, recall the following facts: $E(B_i)=\{1+\bar K(s_i)\}^{-1}\leq \text{$0$$\cdot$$6$}$, $$\text{pr}(\max_{1\leq i \leq \bar N(t)} B_i\leq \text{$0$$\cdot$$6$})=1-\text{pr}(\max_{1\leq i \leq \bar N(t)} B_i> \text{$0$$\cdot$$6$})\geq 1-\bar{N}(t) \text{$0$$\cdot$$4$}^{\bar K(t)},$$ and for any $x\leq \text{$0$$\cdot$$6$}$, $-x-x^2 \leq \log(1-x)\leq -x$. 
Equation \eqref{con} is bounded by
\begin{align}\label{exp}
&\sup_{s\leq t}| \exp\{\sum_{i=1}^{\bar N(s)} \log(1-B_i)\}-\exp\{\sum_{i=1}^{\bar N(s)} \log(1-E(B_i))\}| \nonumber \\
\leq &\sup_{s \leq t} |\exp\{-\sum_{i=1}^{\bar N(s)} B_i\}- \exp\{-\sum_{i=1}^{\bar N(s)} [E(B_i)+\{E(B_i)\}^2]\}|+\sup_{s\leq t} |\exp\{-\sum_{i=1}^{\bar N(s)} (B_i+B_i^2)\}-\exp\{-\sum_{i=1}^{\bar N(s)} E(B_i)\}| \nonumber \\
\leq &\sup_{s \leq t} |\exp\{-\sum_{i=1}^{\bar N(s)} B_i\}-\exp\{-\sum_{i=1}^{\bar N(s)} E(B_i)\}|+ \sup_{s \leq t} |\exp\{-\sum_{i=1}^{\bar N(s)} E(B_i)\}-\exp\{-\sum_{i=1}^{\bar N(s)} [E(B_i)+\{E(B_i)\}^2]\}| \nonumber \\
+& \sup_{s\leq t} |\exp\{-\sum_{i=1}^{\bar N(s)} (B_i+B_i^2)\}-\exp\{-\sum_{i=1}^{\bar N(s)} B_i\}|+\sup_{s\leq t} |\exp\{-\sum_{i=1}^{\bar N(s)} B_i\}-\exp\{-\sum_{i=1}^{\bar N(s)} E(B_i)\}| \nonumber \\
\leq & 2\sup_{s \leq t}|\exp\{-\sum_{i=1}^{\bar N(s)} B_i+E(B_i)-E(B_i)\}-\exp\{-\sum_{i=1}^{\bar N(s)} E(B_i)\}| + \sum_{i=1}^{\bar{N}(t)} \{\bar{K}(t)+1\}^{-2}+ \sum_{i=1}^{\bar{N}(t)} B_i^2 \nonumber \\
\leq & 2\sup_{s\leq t} |\sum_{i=1}^{\bar N(s)} \{B_i-E(B_i)\}|\exp\{-\sum_{i=1}^{\bar N(s)} E(B_i)\}+ \bar{N}(t)/\bar{K}(t)^{-2}+ \sum_{i=1}^{\bar{N}(t)} B_i^2,
\end{align}
with probability larger than $1-\bar{N}(t) \text{$0$$\cdot$$4$}^{\bar K(t)}$. Since $\exp\{-\sum_{i=1}^{\bar N(s)} E(B_i)\}$ is bounded by 1 for any $s  \leq t$, to complete the proof we only need to bound $\sup_{s \leq t}|\sum_{i=1}^{\bar N(s)} \{B_i-E(B_i)\}|$.


Let $T_m=\sum_{i=1}^m  \{B_i-E (B_i)\}$. Then we have $E(T_m)=0$, $\text{var}(T_m) \leq m/\bar K(t)^2 \rightarrow 0$. By Kolmogorov's inequality $\text{pr}(\max_{1\leq m \leq n}| T_m|\geq x)\leq x^{-2} \text{var}(T_n)$ \citep{durrett2010probability}, we know
\begin{align}\label{eq:kom}
\text{pr}(\sup_{s\leq t}|\sum_{i=1}^{\bar N(s)} \{B_i-E(B_i)\}|\geq \epsilon^2/n^{1/2})\leq n\bar{N}(t)/\{\epsilon^2 \bar K(t)\}^2.
\end{align}
Combine \eqref{sumb}, \eqref{exp} and \eqref{eq:kom}, we have
\begin{align*}
\text{pr} \{\sup_{s\leq t}{|S^U(s)-\hat{S}(s)|} \geq 3\epsilon^2/n^{1/2}+\bar{N}(t)/\bar{K}(t)^{-2} \} \leq \bar{N}(t) [ (1-\epsilon/n^{3/4})^{\bar K(t)}+\text{$0$$\cdot$$4$}^{\bar K(t)}+ n/\{\epsilon^2 \bar K(t)\}^2].
\end{align*}
This completes the proof.
\end{proof}

%

For the proof of the next Theorem, we will construct a martingale, and check the two conditions similar to Theorem 5.1.1 in \cite{fleming2011counting}.
\begin{proof}[of Theorem~\ref{main}]
For any $t$ with $\pi(t)>0$, consider a fixed growing sequence of data $(\by,\bdelta)$ for which the statement of Assumption~\ref{ass1} and \ref{ass4} are valid. The set of all such sequences is assumed to have probability one.

For convenience, we denote $\text{pr}^*_{\by,\bdelta}$ as pr in the rest of this section. Additionally, in this proof only, we denote $S^U, F^L$ as $\tilde S, \tilde F$, and define $u(x)=\sum_{i=1}^{\bar N(t)} I\{s_{i-1}<x\leq s_i\} B_i$, where $s_i$ are ordered failure times, $s_0$ is assumed to be 0, and $B_i$ are independent $Beta(1,\bar K(s_i))$. Let 
$\tilde{\Lambda}(s)=\int_0^s u(x) d\bar N(x)=\sum_{i=1}^{\bar N(t)} B_i$. For fixed $t \in \mathcal{I}$, suppose $0\leq s\leq t$, we could rewrite $\tilde S$ recursively as
\begin{align*}
\tilde S(s)=1-\int_0^s \tilde S(x-)d\tilde{\Lambda}(x).
\end{align*}
Then we have
\begin{align*}
\tilde S(s-)-\tilde S(s)=-\Delta \tilde S(s)=\tilde S(s-)\Delta \bar{N}(s)u(s),
\end{align*}
\begin{align*}
\tilde S(s)=\tilde S(s-)\{1-\Delta \bar{N}(s)u(s)\},
\end{align*}
This is the same as Equation~\eqref{pest}.

We know $\hat{S}(s)>0$, therefore
\begin{align*}
\frac{\tilde S(s)}{\hat S(s)}&=\frac{\tilde S(0)}{\hat S(0)}+\int_0^s \tilde S(x-)[-\{\hat S(x) \hat S(x-)\}^{-1}d\hat S(x)]+\int_0^s\frac{1}{\hat S(x)}d\tilde S(x)\\
&=1-\int_0^s \frac{\tilde S(x-)}{\hat S(x)}\{d\bar{N}(x)u(x)-\frac{d\bar{N}(x)}{1+\bar K(x)}\},
\end{align*}
so
\[
\tilde S(s)-\hat S(s)=-\hat S(s) \int_0^s \frac{\tilde S(x-)}{\hat S(x)}\{d\bar{N}(x)u(x)-\frac{d\bar{N}(x)}{1+\bar K(x)}\},
\]
and
\begin{equation}\label{proof1}
n^{1/2}\{\tilde F(s)-\hat F(s)\}=\hat S(s) \int_0^s n^{1/2} \frac{\tilde S(x-)}{\hat S(x)}\{d\bar{N}(x)u(x)-\frac{d\bar{N}(x)}{1+\bar K(x)}\}.
\end{equation}

Now we want to find the asymptotic distribution of right-hand-side of \eqref{proof1}. First, notice that for our fixed sequence of data, $\hat S(s)\to 1-F_0(s)$.  Next, let
\[
U(s)=\int_0^s n^{1/2} \frac{\tilde S(x-)}{\hat S(x)}\{d\bar{N}(x)u(x)-\frac{d\bar{N}(x)}{1+\bar K(x)}\}.
\]
 We need to construct a martingale $M(s)$ to use the martingale central limit theorem. Let
\[
M(s)=\sum_{x \leq s}[u(x)\{1+\bar K(x)\} \{2+\bar K(x)\}^{1/2}- \{2+\bar K(x)\}^{1/2}]\Delta \bar N(x).
\]
It is easy to see that $M(s)$ is a martingale,
\begin{align*}
dM(s)&=0~~ \text{if}~~ \Delta \bar N(s)=0,\\
dM(s)&=u(s)\{1+\bar K(s)\}\{2+\bar K(s)\}^{1/2}-\{2+\bar K(s)\}^{1/2}~~ \text{if}~~ \Delta \bar N(s)=1.
\end{align*}
From here
\[
dM(s)=d\bar N(s)[u(s)\{1+\bar K(s)\}\{2+\bar K(s)\}^{1/2}-\{2+\bar K(s)\}^{1/2}].
\]
Let
\[
H(s)=n^{1/2} \frac{\tilde S(s-)}{\hat S(s)\{1+\bar K(s)\}\{2+\bar K(s)\}^{1/2}},
\]
then
\[
U(s)=\int_0^s H(x)dM(x).
\]
In order to obtain desired convergence, we need to establish the two conditions of Theorem 5.1.1 in \cite{fleming2011counting}.

First, we need to check the first condition
\begin{equation}\label{checkeq1}
<U,U>(s) \stackrel{pr}{\rightarrow} \int_0^s f^2(x)dx, \mbox{ where } f(x)=\{\lambda(x)/\pi(x)\}^{1/2}.
\end{equation}
We have
\[
d<M,M>(x)=\text{var}(dM(x)|\mathcal F_{x-})=\bar{K}(x) d\bar N(x),
\]
and
\[
<U,U>(s)=\int_0^s n \frac{\tilde S^2(x-)\bar K(x)d\bar N(x)}{\hat S^2(x)\{1+\bar K(x)\}^2\{2+\bar K(x)\}}.
\]

By Assumption \ref{ass1}, we have
\[\forall n\geq n_0,\quad  \text{pr}\left(\sup_{x\leq s} \left|\frac{\tilde S^2(x-)}{\hat S^2(x)}\left\{\frac{n}{\bar K(x)}-\frac{1}{\pi(x)}\right\}\right|> \epsilon/3\right)< \epsilon/2.
\]
By the consistency of $\tilde{S}$, we have
\[
\forall n\geq n_1,\quad \text{pr} \left(\sup_{x\leq s} \left|\frac{1}{\pi(x)}\left\{\frac{\tilde S^2(x-)}{\hat S^2(x)}-1\right\}\right|> \epsilon/2\right)< \epsilon/2.
\]
So for $\forall \epsilon >0, \forall n\geq \max(n_0,n_1)$,
\[
\text{pr}\left(\sup_{x\leq s} \left|\frac{\tilde S^2(x-)n}{\hat S^2(x)\bar K(x)}- \frac{1}{\pi(x)}\right|> \epsilon\right)<\epsilon.
\]
Then by Assumption \ref{ass3}, the condition \eqref{checkeq1} is satisfied.

Then we need to check the second condition, i.e., $<U_\epsilon,U_\epsilon>(s) \stackrel{pr}{\rightarrow}0$. For any $\epsilon>0$,
\[
<U_\epsilon,U_\epsilon>(s)=\int_0^s n \frac{\tilde S^2(x-)\bar K(x)d\bar N(x)}{\hat S^2(x)\{1+\bar K(x)\}^2\{2+\bar K(x)\}}I\{  \frac{n^{1/2} \tilde S(x-)}{\hat S(x)\{1+\bar K(x)\}\{2+\bar K(x)\}^{1/2}}\geq \epsilon\}.
\]
Consistency and Assumption \ref{ass1} implies
\begin{align}\label{Hsup1}
&\sup_{x\leq s} \left|H^2(x)\{1+\bar K(x)\}\{2+\bar K(x)\}-\frac{1}{\pi(x)}\right| \nonumber \\
= &\sup_{x\leq s} \left|n\frac{\tilde S^2(x-)}{\hat S^2(x)\{1+\bar K(x)\}}+ \frac{\tilde S^2(x-)}{\hat S^2(x)\pi(x)}-  \frac{\tilde S^2(x-)}{\hat S^2(x)\pi(x)}-\frac{1}{\pi(x)}\right|\nonumber \\
\leq &\frac{1}{\hat S^2(s)}\sup_{x\leq s} \left|\frac{n}{1+\bar K(x)}-\frac{1}{\pi(x)}\right|+\frac{1}{\hat S(s)\pi(s)}\sup_{x\leq s}\left|\tilde S(x-)-\hat S(x)\right| \stackrel{pr}{\rightarrow}0.
\end{align}
From $\bar K(x) \stackrel{pr}{\rightarrow} \infty$ and monotonicity of $\bar K$ we have
\[
\inf_{x\leq s} |\bar K(x)|\stackrel{pr}{\rightarrow}\infty.
\]
Combined with Equation~\eqref{Hsup1}, we have
\begin{align*}
\sup_{x\leq s} |H(x)|\stackrel{pr}{\rightarrow}0,
\end{align*}
which is equivalent to
\begin{align*}
\sup_{x\leq s} I\{  \frac{n^{1/2} \tilde S(x-)}{\hat S(x)\{1+\bar K(x)\} \{2+\bar K(x)\}^{1/2}}\geq \epsilon\} \stackrel{pr}{\rightarrow}0.
\end{align*}
Then
\begin{align*}
\int_0^s n \frac{\tilde S^2 (x-)\bar K(x)d\bar N(x)}{\hat S^2(x)\{1+\bar K(x)\}^2 \{2+\bar K(x)\}}I\{  \frac{n^{1/2} \tilde S(x-)}{\hat S(x)\{1+\bar K(x)\} \{2+\bar K(x)\}^{1/2}}\geq \epsilon\} \stackrel{pr}{\rightarrow}0,
\end{align*}
and the second condition is satisfied.

By replicating the proof in Theorem 5.1.1 in \cite{fleming2011counting} for our martingale, we get $U(s)\Rightarrow U_{\infty}(s)=\int_0^s \{\lambda(x)/\pi(x)\}^{1/2}dW(x)$. We know

\begin{align*}
\text{cov}(U_{\infty}(s_1),U_{\infty}(s_2))=\int_0^{s_1}\frac{\lambda(x)}{\pi(x)}ds=\gamma(s_1)~~\text{for}~~s_1<s_2,
\end{align*}
and
\begin{align*}
\text{cov}(W\{\gamma(s_1)\},W\{\gamma(s_2)\})=\gamma(s_1)~~\text{for}~~s_1<s_2.
\end{align*}
So $U_{\infty}(\cdot)$ is the same as $W\{\gamma(\cdot)\}$. The conclusion of the Theorem \ref{main} follows.
\end{proof}

We conclude this section by proving the corollary.

\begin{proof}[of Corollary~\ref{corollary}]
We know $n^{1/2}\{\hat F(\cdot)-F_0(\cdot)\}\rightarrow \{1-F_0(\cdot)\}W\{\gamma(\cdot)\}$ on $D[0,t]$ and $n^{1/2}\{F^L(\cdot)-\hat F(\cdot)\}\rightarrow \{1-F_0(\cdot)\}W\{\gamma(\cdot)\}$ in distribution on $D[0,t]$ almost surely from Theorem \ref{main}.

From the properties in \eqref{p1} we have that the fiducial probability
\begin{align}\label{continuous}
\notag 1-\alpha&=\text{pr}^*_{\by,\bdelta}(\{F:\Psi\{F(\cdot)-\hat F(\cdot)\}\leq \epsilon_{n,\alpha}\})\\
&=\text{pr}^*_{\by,\bdelta}(\{F:\Psi[ n^{1/2} \{F(\cdot)-\hat F(\cdot)\}]\leq \psi(n^{1/2})\epsilon_{n,\alpha}\}).
\end{align}
By continuous mapping theorem and the fact  that $\Psi[\{1-F_0(\cdot)\}W\{\gamma(\cdot)\}]$ is continuous and has unique $(1-\alpha)$-th quantile, the right-hand side of Equation \eqref{continuous} converges to
\begin{align*}
\text{pr}(\Psi[\{1-F_0(\cdot)\}W\{\gamma(\cdot)\}]\leq \epsilon_\infty),
\end{align*}
where $\epsilon_\infty$ is the unique limit of $\psi(n^{1/2})\epsilon_{n,\alpha}$, and pr is the sampling distribution of the data.

Then we have
\begin{align*}
\text{pr}(F_0 \in \{F: \Psi\{F(\cdot)-\hat F(\cdot)\}\leq \epsilon_{n,\alpha} \} )&= \text{pr}( \Psi\{F_0(\cdot)-\hat F(\cdot)\}\leq \epsilon_{n,\alpha}  )\\
&=\text{pr}(\Psi[n^{1/2}\{F_0(\cdot)-\hat F(\cdot)\}]\leq \psi(n^{1/2})\epsilon_{n,\alpha}  )\\
&\rightarrow \text{pr}(\Psi[\{1-F_0(\cdot)\}W\{\gamma(\cdot)\}]\leq \epsilon_\infty)\\
&=1-\alpha.
\end{align*}

This completes the proof.
\end{proof}



\section*{Results for alternative selection schemes}\label{a:lemma}
\begin{lemma}\label{a:lemma1}
The following modification of Theorem \ref{consistency} is valid for $S^L$:
{\small
\begin{multline}\label{eq:concentration2}
\text{pr}^*_{\by,\bdelta} \{\sup_{s\leq t}{|S^L(s)-\hat{S}(s)|} \geq \epsilon/n^{3/4}+3\epsilon^2/n^{1/2}+\bar{N}(t)/\bar{K}(t)^{-2} \} \\
\leq \{\bar{N}(t)+1\}(1-\epsilon/n^{3/4})^{\bar K(t)}+\bar{N}(t) [\text{$0$$\cdot$$4$}^{\bar K(t)}+ n/\{\epsilon^2 \bar K(t)\}^2].
\end{multline}
}
The same bound also holds for $S^I$. Moreover, Theorem 
 \ref{main} holds for $S^L$ and $S^I$.
\end{lemma}
\begin{proof}
Recall that $S^L(s)\geq S^U(s^+)$ and $S^L(s)\leq S^U(s)$ hold for any $s\leq t$, where $s^+$ denotes the next failure time right after $s$. 
Furthermore, the difference between $S^U(s)$ and $S^U(s^+)$ is bounded by
\begin{align*}
| S^U(s)-S^U(s^+)|  = &|\prod_{i=1}^{\bar N(s)} \{1- B_i\} - \prod_{i=1}^{\bar N(s)} \{1- B_i\}(1- B_{\bar N(s)+1})|\nonumber \\
 = &|\prod_{i=1}^{\bar N(s)} \{1- B_i\}  B_{\bar N(s)+1}|
\leq \max_{1\leq i\leq \bar N(s)+1} B_i,
\end{align*}
where $B_i$ follows $Beta(1,\bar K(s_i))$ and $s_i$ are ordered failure times before or at time $s$. By Equation \eqref{a3} in the previous section, we have
\begin{align*}
\text{pr}^*_{\by,\bdelta}( | S^U(s^+)-S^U(s)|> \epsilon/n^{3/4})\leq \text{pr}^*_{\by,\bdelta}( \max_{1\leq i\leq \bar N(s)+1} B_i >\epsilon/n^{3/4} )\leq \{\bar{N}(t)+1\} (1-\frac{\epsilon}{n^{3/4}})^{\bar K(t)}. 
\end{align*}
Notice that $S^L(s)-\hat S(s)= \{S^L(s)-S^U(s)\}+ \{S^U(s)-\hat S(s)\}$ and $S^U(s^+)-S^U(s)\leq S^L(s)-S^U(s) \leq 0$. This implies \eqref{eq:concentration2}. In addition, Theorem \ref{main} holds for $S^L$ and $S^I$ by Slutsky's theorem.
%
\end{proof}

\begin{lemma}\label{a:lemma2}
 For any failure time $t$, $E^*_{\by,\bdelta}[S^L(t)]\leq \tilde S(t)\leq E^*_{\by,\bdelta}[S^U(t)]$, where $E^*_{\by,\bdelta}$ is the expectation with respect to $\text{pr}^*_{\by,\bdelta}$, and $\tilde S(t)$ is the Kaplan-Meier estimator.
\end{lemma}
\begin{proof}
For any failure time $t$, we have $S^U(t)=\prod_{i=1}^{\bar N(t)} \left \{ 1- B_i \right \}$.
From here
\begin{equation*}
E^*_{\by,\bdelta} [S^U(t)]=\prod_{i=1}^{\bar N(t)} \left \{ 1-\frac{1}{1+\bar K(s_i)} \right \} \geq \prod_{i=1}^{\bar N(t)} \left \{ 1-\frac{1}{\bar K(s_i)} \right \},
\end{equation*}
where $s_i$ are ordered failure times. Similarly, $S^L(t)=S^U(t)(1-B)$,
where $B$ follows $Beta(1,\bar K(t)-1)$ and is independent of $B_i$ for $i \leq \bar N(t)$. Thus
\begin{equation*}
E^*_{\by,\bdelta} [S^L(t)]=\prod_{i=1}^{\bar N(t)} \left \{ 1-\frac{1}{1+\bar K(s_i)} \right \}(1-\frac{1}{\bar K(t)}) \leq \prod_{i=1}^{\bar N(t)} \left \{ 1-\frac{1}{\bar K(s_i)} \right \}.
\end{equation*}
This completes the proof.
\end{proof}

%
%
%
%
%
%
%
%
%
%
%
%

\section*{Algorithm for sampling from the fiducial distribution}\label{al2}

1. Generate $U=(u_1,\ldots,u_n)$ from $U(0,1)$ and sort them. Denote sorted values as $preU$.

2. Sort the data. Denote sorted data as $(y_1,\ldots,y_n)$ and $(\delta_1,\ldots,\delta_n)$. 

3. Initialize $LowerFid=(0)_{n+1}$, $UpperFid=(1)_{n+1}$.

4. For $i=1$ to $n$:

Let $UpperFid(i)=preU(1)$, where $preU(1)$ is the smallest element left in $preU$.

If $\delta=1$, set $LowerFid(i+1)=preU(1)$, and delete $preU(1)$;

If $\delta=0$, randomly pick one $u$ from $preU$, set $LowerFid(i+1)=LowerFid(i)$, and delete the selected $u$ from $preU$.

5. We output 3 survival functions that are needed for the conservative and log-linear interpolation methods.

5.1. Lower fiducial bound: using $LowerFid$ as a fiducial curve.

5.2. Upper fiducial bound: using $UpperFid$ as a fiducial curve.


5.3. Log-linear interpolation: Fit a continuous fiducial distribution by linear interpolation based on failure observations as described in Section \ref{nonsur}. Then correct the linear interpolation at the censoring observations so that the upper fiducial bound on continuous distribution function (lower fiducial bound for survival function) is satisfied. Let $y_{n-k}$ ($k=0,1,\ldots,n-1$) denotes the last failure observation. We fit a single line after last uncensored observation and take the maximum of  $s_0,s_1,\ldots,s_{k}$ as slope, where $s_1$ is the slope between $(y_{n-k},\log u_{n-k})$ and $(y_{n-k+1},\log u_{n-k+1}),$ $\ldots$, $s_{k}$ is the slope between $(y_{n-k},\log u_{n-k})$ and $(y_{n},\log u_{n})$, $s_0$ is the slope between $(\tilde y,\log \tilde u)$ and $(y_{n-k},\log u_{n-k})$, $\tilde y$ is the second last uncensored observation. If there is only one failure time, $\tilde y$ and $\log \tilde u$ are 0.

6. From step 1--5 we get one curve of fiducial distribution. Repeat step 1--5 to get one fiducial sample with $m$ curves.

\bibliographystyle{rss}
\bibliography{Bibliography}

\end{document}